\documentclass[10pt,a4paper]{article}
\usepackage[a4paper]{geometry}
\usepackage{amssymb,latexsym,amsmath,amsfonts,amsthm}
\usepackage{graphicx}

\usepackage{epsfig}

\DeclareMathOperator{\diag}{diag} 

\newcommand{\er}{\mathbb{R}}
\newcommand{\cee}{\mathbb{C}}
\newcommand{\enn}{\mathbb{N}}
\newcommand{\lam}{\lambda}

\newcommand{\til}{\tilde}
\newcommand{\K}{\widetilde K}
\renewcommand{\P}{\widetilde P}
\newcommand{\Q}{\widetilde Q}

\newcommand{\bol}{\hfill\square\\}

\newtheorem{theorem}{Theorem}[section]
\newtheorem{lemma}[theorem]{Lemma}
\newtheorem{proposition}[theorem]{Proposition}

\newtheorem{rhp}[theorem]{RH problem}

\theoremstyle{definition}

\theoremstyle{remark}

\newtheorem{remark}[theorem]{Remark}

\numberwithin{equation}{section}

\title{Average characteristic polynomials in the two-matrix model}

\author{Steven Delvaux\footnotemark[1]}
\date{\today}

\begin{document}

\maketitle
\renewcommand{\thefootnote}{\fnsymbol{footnote}}
\footnotetext[1]{Department of Mathematics, Katholieke Universiteit Leuven,
Celestijnenlaan 200B, B-3001 Leuven, Belgium. email:
steven.delvaux\symbol{'100}wis.kuleuven.be. The author is a Postdoctoral Fellow
of the Fund for Scientific Research - Flanders (Belgium).}

\begin{abstract}
The two-matrix model is defined on pairs of Hermitian matrices $(M_1,M_2)$ of
size $n\times n$ by the probability measure \begin{equation*}\frac{1}{Z_n}
\exp\left(\textrm{Tr} (-V(M_1)-W(M_2)+\tau M_1M_2)\right)\ dM_1\ dM_2,
\end{equation*}
where $V$ and $W$ are given potential functions and $\tau\in\er$. We study
averages of products and ratios of characteristic polynomials in the two-matrix
model, where both matrices $M_1$ and $M_2$ may appear in a combined way in both
numerator and denominator. We obtain determinantal expressions for such
averages. The determinants are constructed from several building blocks: the
biorthogonal polynomials $p_n(x)$ and $q_n(y)$ associated to the two-matrix
model; certain transformed functions $\P_n(w)$ and $\Q_n(v)$; and finally
Cauchy-type transforms of the four Eynard-Mehta kernels $K_{1,1}$, $K_{1,2}$,
$K_{2,1}$ and $K_{2,2}$. In this way we generalize known results for the
$1$-matrix model. Our results also imply a new proof of the Eynard-Mehta
theorem for correlation functions in the two-matrix model, and they lead to a
generating function for averages of products of traces.


\end{abstract}

\section{Introduction and statement of results}
\label{section:introduction}

\subsection{The two-matrix model}
\label{subsection:twomatrixdef}

Let $\mathcal H_n$ be the space of Hermitian $n$ by $n$ matrices. The
\emph{two-matrix model} is defined on $\mathcal H_n\times \mathcal H_n$ by the
probability measure
\begin{equation}\label{twomatrixmodel}
\frac{1}{Z_n} \exp\left(\textrm{Tr} (-V(M_1)-W(M_2)+\tau M_1M_2)\right)\ dM_1\
dM_2,\qquad (M_1,M_2)\in\mathcal H_n\times \mathcal H_n,
\end{equation}
where $V$ and $W$ are given real-valued \emph{potential functions}, where
$\tau\in\er$ is the \emph{coupling constant}, $Z_n$ is a normalization
constant, $\textrm{Tr}$ denotes the trace and $dM_1 dM_2$ is the Lebesgue
measure on $\mathcal H_n\times \mathcal H_n$, see \cite{Mehta}. The matrix
model \eqref{twomatrixmodel} is probably the most popular instance of a
two-matrix model in the literature, see e.g.\
\cite{AvMcoupled,BEy,BEH1,BEH2,BHI,DiFran,Duits2,EMcL,EM,EO,Kap,KMcL,Mehta,Mehta2}
among many others. We may mention however that other kinds of two-matrix models
exist as well, see e.g.\ \cite{BGS} and the references therein.

The normalization constant $Z_n$ in \eqref{twomatrixmodel} is called the
\emph{partition function}. It is defined as the total integral
\begin{equation}\label{partitionfunction} Z_n = \int_{\mathcal H_n}\!\int_{\mathcal H_n}\!
\exp\left(\textrm{Tr}(-V(M_1)-W(M_2)+\tau M_1M_2)\right)\ dM_1\ dM_2.
\end{equation}
We assume that the potential functions $V$ and $W$ have sufficient increase at
infinity in order that the integral \eqref{partitionfunction} is convergent.
Throughout this paper, we will assume that $V$ and $W$ are polynomials of even
degree with positive leading coefficients.

Due to the \emph{Harish-Chandra formula} (or \emph{Itzykson-Zuber formula})
\cite{Mehta}, integrals over the two-matrix model \eqref{twomatrixmodel} can
often be rewritten in terms of the eigenvalues
of $M_1$ and $M_2$. We need the following version of this result. Assume that
the functions $f,g:\mathcal H_n\to\cee$ depend only on the eigenvalues
$(\lam_i)_{i=1}^n$ of $M_1$ and $(\mu_i)_{i=1}^n$ of $M_2$, respectively,
(taking into account eigenvalue multiplicities), in the sense that
$$ f(M_1) = \til f(\lam_1,\ldots,\lam_n),\qquad g(M_2) = \til
g(\mu_1,\ldots,\mu_n),
$$
where $\til f,\til g:\er^n\to\cee$ are symmetric functions of $n$ variables.
Then
\begin{multline}\label{inteigs}
\frac{1}{Z_n}\int_{\mathcal H_n}\!\int_{\mathcal H_n}\! f(M_1)g(M_2)
e^{\textrm{Tr} (-V(M_1)-W(M_2)+\tau M_1M_2)}\ dM_1\ dM_2 \\ = \frac{1}{\til
Z_n}\int_{-\infty}^{\infty}\!\ldots\int_{-\infty}^{\infty}\! \til
f(\lam_1,\ldots,\lam_n)\til g(\mu_1,\ldots,\mu_n)\prod_{i=1}^n \left(
e^{-V(\lam_i)} e^{-W(\mu_i)}\right)\\
\times\Delta(\mathbf{\lam})\Delta(\mathbf{\mu})\det(e^{\tau\lam_i\mu_j})_{i,j=1}^n
\prod_{i=1}^n \left( d\lam_i\ d\mu_i\right),
\end{multline}
for a new normalization constant $\til Z_n$, see \cite{Mehta}. The factors
$\Delta(\mathbf{\lam})$ and $\Delta(\mathbf{\mu})$ in \eqref{inteigs} denote
the Vandermonde determinants $\prod_{1\leq i<j\leq n}(\lam_j-\lam_i)$ and
$\prod_{1\leq i<j\leq n}(\mu_j-\mu_i)$ respectively. In \eqref{inteigs} we also
assume that $\til f$ and $\til g$ are such that the integrals converge.

\subsection{Biorthogonal polynomials and kernels}
\label{subsection:kernelsdefs}

It is known that the two-matrix model \eqref{twomatrixmodel} can be analyzed by
means of biorthogonal polynomials and kernel functions that we describe below
\cite{EMcL,EM,Mehta}. We will also introduce certain (non-standard) Cauchy-type
transforms of these objects, which will be needed in the statement of our main
theorems.

The monic \emph{biorthogonal polynomials} $\{p_i(x)\}_{i=0}^{\infty}$ and
$\{q_j(y)\}_{j=0}^{\infty}$ are defined as the monic polynomials
$$p_i(x)=x^i+O(x^{i-1}),\qquad q_j(y)=y^j+O(y^{j-1}), $$
that satisfy the orthogonality relations
\begin{equation}\label{def:pnqn}
\int_{-\infty}^{\infty}\!\int_{-\infty}^{\infty}\!
p_i(x)q_j(y)e^{-V(x)-W(y)+\tau xy}\ dx\ dy=0,\qquad i\neq j.
\end{equation}
These polynomials exist and are unique, see \cite{EMcL}.

We define the \emph{transformed functions} $P_i$ and $Q_j$ by
\begin{eqnarray}\label{def:Pn}
P_i(y) &=& \int_{-\infty}^{\infty}\! p_i(x) e^{-V(x)-W(y)+\tau xy}\ dx,
\\ \label{def:Qn}
Q_j(x) &=& \int_{-\infty}^{\infty}\! q_j(y) e^{-V(x)-W(y)+\tau xy}\ dy,
\end{eqnarray}
and we let $h_i^2$ be defined by
\begin{equation}\label{def:hk}
h_i^2 = \int_{-\infty}^{\infty}\!\int_{-\infty}^{\infty}\!
p_i(x)q_i(y)e^{-V(x)-W(y)+\tau xy}\ dx\ dy.
\end{equation}
The integral \eqref{def:hk} is positive, see \cite{EMcL}.

The \emph{Eynard-Mehta kernels} $K_{1,1}$, $K_{1,2}$, $K_{2,1}$ and $K_{2,2}$
are defined as \cite{EM,BR}
\begin{eqnarray}\label{def:K11}
K_{1,1}(x_1,x_2) &=& \sum_{i=0}^{n-1} \frac{1}{h_i^2}p_i(x_1)Q_i(x_2),
\\ \label{def:K12}
K_{1,2}(x,y) &=& \sum_{i=0}^{n-1} \frac{1}{h_i^2}p_i(x)q_i(y),
\\ \label{def:K21}
K_{2,1}(y,x) &=& \left(\sum_{i=0}^{n-1}
\frac{1}{h_i^2}P_i(y)Q_i(x)\right)-e^{-V(x)-W(y)+\tau xy},
\\ \label{def:K22}
K_{2,2}(y_1,y_2) &=& \sum_{i=0}^{n-1} \frac{1}{h_i^2}P_i(y_1)q_i(y_2).
\end{eqnarray}
It is known that the correlation functions for the two-matrix model
\eqref{twomatrixmodel} can be written as determinants of matrices built out of
the four Eynard-Mehta kernels; see \cite{EM} and see also
Section~\ref{subsection:EynardMehta} below.

In what follows, we will need certain Cauchy-type transforms of the above
objects. First we define the Cauchy transforms $\P_i(w)$ and $\Q_j(v)$ of the
functions $P_i$ and $Q_j$ in \eqref{def:Pn}--\eqref{def:Qn}:
\begin{eqnarray}\label{def:Pn:til}
\P_i(w) &=& \int_{-\infty}^{\infty}\! \frac{P_i(\eta)}{w-\eta}\ d\eta =
\int_{-\infty}^{\infty}\!\int_{-\infty}^{\infty}\!
\frac{p_i(\xi)e^{-V(\xi)-W(\eta)+\tau\xi\eta}}{w-\eta}\ d\xi\ d\eta,
\\
\label{def:Qn:til} \Q_j(v) &=& \int_{-\infty}^{\infty}\!
\frac{Q_j(\xi)}{v-\xi}\ d\xi =
\int_{-\infty}^{\infty}\!\int_{-\infty}^{\infty}\!
\frac{q_j(\eta)e^{-V(\xi)-W(\eta)+\tau\xi\eta}}{v-\xi}\ d\xi\ d\eta.
\end{eqnarray}
We will also need the following Cauchy-type transforms of the Eynard-Mehta
kernels:
\begin{eqnarray}\label{def:K11til}
\K_{1,1}(x,v) &=& \frac{1}{x-v}\int_{-\infty}^{\infty}
\frac{x-\xi}{v-\xi}K_{1,1}(x,\xi)\ d\xi,
\\ \label{def:K12til}
\K_{1,2}(x,y) &=& K_{1,2}(x,y),
\\ \label{def:K21til}
\K_{2,1}(w,v) &=&
\int_{-\infty}^{\infty}\!\int_{-\infty}^{\infty}\frac{K_{2,1}(\eta,\xi)}{(w-\eta)(v-\xi)}\
d\xi\ d\eta,
\\  \label{def:K22til}
\K_{2,2}(w,y) &=& \frac{1}{y-w}\int_{-\infty}^{\infty}
\frac{y-\eta}{w-\eta}K_{2,2}(\eta,y)\ d\eta.
\end{eqnarray}
See Section~\ref{subsection:reproducing} for several properties of the kernels
$\K_{i,j}$, $i,j=1,2$. Note that in the above formulas, we consistently use \lq
$\xi$\rq\ and \lq $\eta$\rq\ to denote integration variables which play the
role of an \lq $x$-variable\rq\ or a \lq$y$-variable\rq\ respectively.

The following lemma gives some alternative expressions for the kernels
$\K_{i,j}$, $i,j=1,2$.

\begin{lemma}\label{lemma:sumformulas} (Summation formulas for the kernels $\K_{i,j}$:) We have the following
equivalent expressions for the kernels $\K_{i,j}$ in
\eqref{def:K11til}--\eqref{def:K22til}:
\begin{eqnarray}\label{def:K11sum}
\K_{1,1}(x,v) &=&
\left(\sum_{i=0}^{n-1}\frac{1}{h_i^2}p_i(x)\Q_i(v)\right)-\frac{1}{v-x},
\\ \label{def:K12sum}
\K_{1,2}(x,y) &=& \sum_{i=0}^{n-1} \frac{1}{h_i^2}p_i(x)q_i(y),
\\ \label{def:K21sum}
\K_{2,1}(w,v) &=& \left(\sum_{i=0}^{n-1}\frac{1}{h_i^2}\P_i(w)\Q_i(v)\right)-
\int_{-\infty}^{\infty}\!\int_{-\infty}^{\infty}\frac{e^{-V(\xi)-W(\eta)+\tau\xi\eta}}{(w-\eta)(v-\xi)}\
d\xi\ d\eta,
\\  \label{def:K22sum}
\K_{2,2}(w,y) &=&
\left(\sum_{i=0}^{n-1}\frac{1}{h_i^2}\P_i(w)q_i(y)\right)-\frac{1}{w-y},
\end{eqnarray}
where we use the notations in \eqref{def:Pn:til}--\eqref{def:Qn:til}.
\end{lemma}

Lemma~\ref{lemma:sumformulas} is proved in
Section~\ref{subsection:proof:lemma:sum}.

\subsection{Average characteristic polynomials}
\label{subsection:statementresults}

Let the matrices $M_1,M_2\in\mathcal H_n$ be distributed according to the
two-matrix model \eqref{twomatrixmodel}. The \emph{average characteristic
polynomial} of the matrix $M_1$ is defined by (the notation in the left hand
side is explained below)
\begin{equation}\label{avcharpol:00}
P^{[1,0,0,0]}(x):=\frac{1}{Z_n} \int_{\mathcal H_n}\!\int_{\mathcal H_n}\!
\det(x I_n-M_1)\ e^{\textrm{Tr}( -V(M_1)-W(M_2)+\tau M_1M_2)}\ dM_1\ dM_2,
\end{equation}
where $I_n$ denotes the identity matrix of size $n$. Similarly, the average
characteristic polynomial of the matrix $M_2$ is defined by
\begin{equation}\label{avcharpol:01}
P^{[0,1,0,0]}(y):=\frac{1}{Z_n} \int_{\mathcal H_n}\!\int_{\mathcal H_n}\!
\det(y I_n-M_2)\ e^{\textrm{Tr}( -V(M_1)-W(M_2)+\tau M_1M_2)}\ dM_1\ dM_2.
\end{equation}

In this paper, we consider more general versions of
\eqref{avcharpol:00}--\eqref{avcharpol:01}. We are interested in averages of
products and ratios of characteristic polynomials, where both matrices $M_1$
and $M_2$ may appear in a combined way in both numerator and denominator. Our
object of study is of the form
\begin{multline}\label{avcharpol}
P_{n}^{[I,J,K,L]}(x_1,\ldots,x_I;y_1,\ldots,y_J;v_1,\ldots,v_K;w_1,\ldots,w_L):=\\
\frac{1}{Z_n} \int_{\mathcal H_n}\!\int_{\mathcal H_n}\!\frac{\prod_{i=1}^I
\det(x_i I_n-M_1)\prod_{j=1}^J \det(y_j I_n-M_2)}{\prod_{k=1}^K \det(v_k
I_n-M_1)\prod_{l=1}^L \det(w_l I_n-M_2)}\ e^{\textrm{Tr}( -V(M_1)-W(M_2)+\tau
M_1M_2)}\ dM_1\ dM_2.
\end{multline}
Here $I,J,K,L\in\enn\cup\{0\}$ are nonnegative integers and we assume that
$x_1,\ldots,x_I,y_1,\ldots,y_J\in\cee$, that
$v_1,\ldots,v_K,w_1,\ldots,w_L\in\cee\setminus\er$, that the numbers in the set
$(x_1,\ldots,x_I$, $v_1,\ldots,v_K$) are pairwise distinct, and similarly the
numbers in the set $(y_1,\ldots,y_J$, $w_1,\ldots,w_L)$ are pairwise distinct.
The numbers $x_i,y_j,v_k,w_l$ are sometimes called \emph{external
sources}~\cite{Bergere2}.

Due to \eqref{inteigs}, the expression \eqref{avcharpol} can be written
alternatively as
\begin{multline}\label{avcharpol:Harish}
P_{n}^{[I,J,K,L]}(x_1,\ldots,x_I;y_1,\ldots,y_J;v_1,\ldots,v_K;w_1,\ldots,w_L)
\\ = \frac{1}{\til Z_n}\int_{-\infty}^{\infty}\!\ldots\int_{-\infty}^{\infty}
\prod_{m=1}^n\left( \frac{\prod_{i=1}^I (x_i-\lam_m)\prod_{j=1}^J
(y_j-\mu_m)}{\prod_{k=1}^K (v_k-\lam_m)\prod_{l=1}^L (w_l-\mu_m)}\right)\\
\times\prod_{i=1}^n  \left( e^{-V(\lam_i)}
e^{-W(\mu_i)}\right)\Delta(\mathbf{\lam})\Delta(\mathbf{\mu})\det(e^{\tau\lam_i\mu_j})_{i,j=1}^n
\prod_{i=1}^n \left( d\lam_i\ d\mu_i\right).
\end{multline}

In principle, an expression for \eqref{avcharpol} can be obtained as follows:
if we define the \emph{perturbed} potentials $$\widehat
V(x)=V(x)-\log\left(\frac{\prod_{i=1}^I (x_i-x)} {\prod_{k=1}^K
(v_j-x)}\right), \qquad \widehat W(y)=W(y)-\log\left(\frac{\prod_{j=1}^J
(y_j-y)} {\prod_{l=1}^L (w_l-y)}\right),$$ then \eqref{avcharpol} is just the
scaled partition function corresponding to these perturbed potentials:
\begin{equation}
P_{n}^{[I,J,K,L]}(x_1,\ldots,x_I;y_1,\ldots,y_J;v_1,\ldots,v_K;w_1,\ldots,w_L)
= \frac{\widehat Z_n}{Z_n},
\end{equation}
with
$$ \widehat Z_n :=
\int_{\mathcal H_n}\!\int_{\mathcal H_n}\! \exp\left(\textrm{Tr} (-\widehat
V(M_1)-\widehat W(M_2)+\tau M_1M_2)\right)\ dM_1\ dM_2.
$$
Our interest, however, is to express \eqref{avcharpol} in terms of the
\emph{unperturbed} potentials $V$ and $W$. In this way we will obtain analogues
of known results for the $1$-matrix model, which go back to Br\'ezin-Hikami
\cite{BH3}, Fyodorov-Strahov \cite{FS,SF} and others
\cite{BDS,Bergere2,BS,KS,KG}.

To start with, we consider the case where only one external source is present.

\begin{theorem}\label{theorem:avpol1} (Average characteristic polynomials of $M_1$ and $M_2$:)
For the average characteristic polynomials in
\eqref{avcharpol:00}--\eqref{avcharpol:01}, we have
$$ P_{n}^{[1,0,0,0]}(x) = p_n(x),
$$
$$ P_{n}^{[0,1,0,0]}(y) = q_n(y),
$$
where $p_n$ and $q_n$ denote the biorthogonal polynomials in \eqref{def:pnqn}.
\end{theorem}

\begin{theorem}\label{theorem:avpol2} (Average inverse characteristic polynomials of $M_1$ and $M_2$:)
With $(I,J,K,L)$ equal to $(0,0,1,0)$ or $(0,0,0,1)$ in \eqref{avcharpol}, we
have
$$ P_{n}^{[0,0,1,0]}(v) = \frac{\Q_{n-1}(v)}{h_{n-1}^2},
$$
$$ P_{n}^{[0,0,0,1]}(w) = \frac{\P_{n-1}(w)}{h_{n-1}^2},
$$
where we use the notations in \eqref{def:hk} and
\eqref{def:Pn:til}--\eqref{def:Qn:til}.
\end{theorem}

Theorems~\ref{theorem:avpol1} and \ref{theorem:avpol2} will be proved in
Section~\ref{subsection:proof:M1}. We note that the proofs will be quite
different from the proofs of the corresponding statements for the Hermitian
$1$-matrix model, see \cite{BDS,BS,BH3,FS,SF}. The latter proofs typically use
a Heine-type representation of the orthogonal polynomials on the real line in
terms of a moment (Hankel) determinant \cite{Sz}. Such formulas also exists in
the biorthogonal case \cite{EMcL}, but it seems that a straightforward
generalization of the arguments for the $1$-matrix case is not available.
Instead, our proofs of Theorems~\ref{theorem:avpol1} and \ref{theorem:avpol2}
will rely on the \emph{multiple} orthogonality properties of $p_n$ and $q_n$
due to Kuijlaars-McLaughlin \cite{KMcL}.

Next we consider the case where there are exactly two external sources in
\eqref{avcharpol}.

\begin{theorem}\label{theorem:avpol:K} (Two external sources:)
With $(I,J,K,L)$ equal to $(1,0,1,0)$, $(1,1,0,0)$, $(0,0,1,1)$ or $(0,1,0,1)$
in \eqref{avcharpol}, we have \begin{eqnarray*} P_{n}^{[1,0,1,0]}(x,v) &=&
(x-v)\K_{1,1}(x,v),
\\
 P_{n}^{[1,1,0,0]}(x,y) &=& h_{n}^2\K_{1,2}(x,y),\\
 P_{n}^{[0,0,1,1]}(v,w) &=& -h_{n-1}^{-2}\K_{2,1}(w,v),\\
 P_{n}^{[0,1,0,1]}(y,w) &=& (y-w)\K_{2,2}(w,y),
\end{eqnarray*}
where we use the notations in \eqref{def:hk} and
\eqref{def:K11sum}--\eqref{def:K22sum}, but with the index $n$ in
\eqref{def:K12sum} and \eqref{def:K21sum} replaced by $n+1$ and $n-1$
respectively.
\end{theorem}

Theorem~\ref{theorem:avpol:K} is a consequence of the more general
Theorem~\ref{theorem:avpol:gen} below.

Finally we consider the general case in \eqref{avcharpol}. We first state the
result when all external sources are in the numerator.

\begin{theorem}\label{theorem:avpol:num} (External sources in numerator only:)
Assume that $K=L=0$ in \eqref{avcharpol}, i.e., \eqref{avcharpol} has external
sources in the numerator only. Also suppose that $I\geq J$. Then
\begin{multline}\label{matrix:num}
P_n^{[I,J,0,0]}(x_1,\ldots,x_I;y_1,\ldots,y_J) =
\frac{\prod_{i=0}^{J-1}h_{n+i}^2}{\prod_{1\leq
i<j\leq I}(x_j-x_i)\prod_{1\leq i<j\leq J}(y_j-y_i)}\\
\times\det\begin{pmatrix}
\K_{1,2}(x_1,y_1) & \ldots & \K_{1,2}(x_1,y_J) & p_{n+J}(x_{1}) & \ldots & p_{n+I-1}(x_{1})\\
\vdots & & \vdots & \vdots & & \vdots\\
\K_{1,2}(x_I,y_1) & \ldots & \K_{1,2}(x_I,y_J)& p_{n+J}(x_{I}) & \ldots &
p_{n+I-1}(x_{I})
\end{pmatrix}.
\end{multline}
Here  $\K_{1,2}$ is as defined in \eqref{def:K12sum} but with $n$ replaced by
$n+p$ where $p\in\{J,\ldots,I\}$ may be any fixed number.
\end{theorem}

Theorem~\ref{theorem:avpol:num} will be proved in
Section~\ref{subsection:proof:num}. The proof will be based on formulas of
Christoffel type \cite{BDS,Sz}. A close variant to
Theorem~\ref{theorem:avpol:num} was obtained by Akemann et al. in \cite{ADOS}.

In the statement of Theorem~\ref{theorem:avpol:num} we assume that the kernel
$\K_{1,2}$ is defined in \eqref{def:K12sum} but with $n$ replaced by $n+p$ for
an arbitrary but fixed number $p\in\{J,\ldots,I\}$. The fact that all values of
$p\in\{J,\ldots,I\}$ give the same determinant can be seen by applying
elementary column operations to \eqref{matrix:num}. More precisely, to the
$j$th column of \eqref{matrix:num}, $j\in\{1,\ldots,J\}$, we may add a suitable
linear combination of the last $I-J$ columns.

Now we turn to the general case.

\begin{theorem}\label{theorem:avpol:gen} (The general case:)
Suppose that $I-K\geq J-L\geq -n$. Then
\begin{multline}\label{gencase:a}
P_n^{[I,J,K,L]}(x_1,\ldots,x_I;y_1,\ldots,y_J;v_1,\ldots,v_K;w_1,\ldots,w_L) \\
= C\frac{\prod_{i,k}(x_i-v_k)}{\prod_{1\leq i<j\leq I}(x_j-x_i)\prod_{1\leq
k<l\leq K}(v_k-v_l) }\frac{\prod_{j,l}(y_j-w_l)}{\prod_{1\leq i<j\leq
J}(y_j-y_i)\prod_{1\leq k<l\leq L}(w_k-w_l)}\\
\times \det\begin{pmatrix}
\K_{1,1}(x_1,v_1) & \ldots & \K_{1,1}(x_I,v_1) & \K_{2,1}(w_1,v_1) & \ldots & \K_{2,1}(w_L,v_1) \\
\vdots & & \vdots & \vdots & & \vdots\\
\K_{1,1}(x_1,v_K) & \ldots & \K_{1,1}(x_I,v_K) & \K_{2,1}(w_1,v_K) & \ldots & \K_{2,1}(w_L,v_K) \\
\K_{1,2}(x_1,y_1) & \ldots & \K_{1,2}(x_I,y_1) & \K_{2,2}(w_1,y_1) & \ldots & \K_{2,2}(w_L,y_1) \\
\vdots & & \vdots & \vdots & & \vdots \\
\K_{1,2}(x_1,y_J) & \ldots & \K_{1,2}(x_I,y_J) & \K_{2,2}(w_1,y_J) & \ldots & \K_{2,2}(w_L,y_J) \\
p_{n+J-L}(x_1) & \ldots & p_{n+J-L}(x_I) & \P_{n+J-L}(w_1) & \ldots & \P_{n+J-L}(w_L) \\
\vdots & & \vdots & \vdots & & \vdots \\
p_{n+I-K-1}(x_1) & \ldots & p_{n+I-K-1}(x_I) & \P_{n+I-K-1}(w_1) & \ldots &
\P_{n+I-K-1}(w_L)
\end{pmatrix},
\end{multline}
with $$ C = \left\{\begin{array}{ll} (-1)^{(I+K)L}\prod_{i=0}^{J-L-1}
h^2_{n+i},&\quad\textrm{if $J-L\geq 0$},\\
\vphantom{\frac{1^1_1}{1^1_1}}(-1)^{(I+K)L}\prod_{i=J-L}^{-1}
h^{-2}_{n+i},&\quad\textrm{if $J-L< 0$}.
\end{array}\right.$$
Here the kernels $\K_{i,j}$ are as defined in
\eqref{def:K11sum}--\eqref{def:K22sum} but with $n$ replaced by $n+p$ where
$p\in\{J-L,\ldots,I-K\}$ may be any fixed number.

Similarly, if $J-L\geq I-K\geq -n$ then
\begin{multline}\label{gencase:b}
P_n^{[I,J,K,L]}(x_1,\ldots,x_I;y_1,\ldots,y_J;v_1,\ldots,v_K;w_1,\ldots,w_L) \\
= C\frac{\prod_{i,k}(x_i-v_k)}{\prod_{1\leq i<j\leq I}(x_j-x_i)\prod_{1\leq
k<l\leq K}(v_k-v_l)
}\frac{\prod_{j,l}(y_j-w_l)}{\prod_{1\leq i<j\leq J}(y_j-y_i)\prod_{1\leq k<l\leq L}(w_k-w_l)}\\
\times \det\begin{pmatrix}
\K_{1,1}(x_1,v_1) & \ldots & \K_{1,1}(x_1,v_K) & \K_{1,2}(x_1,y_1) & \ldots & \K_{1,2}(x_1,y_J) \\
\vdots & & \vdots & \vdots & & \vdots\\
\K_{1,1}(x_I,v_1) & \ldots & \K_{1,1}(x_I,v_K) & \K_{1,2}(x_I,y_1) & \ldots & \K_{1,2}(x_I,y_J) \\
\K_{2,1}(w_1,v_1) & \ldots & \K_{2,1}(w_1,v_K) & \K_{2,2}(w_1,y_1) & \ldots & \K_{2,2}(w_1,y_J) \\
\vdots & & \vdots & \vdots & & \vdots \\
\K_{2,1}(w_L,v_1) & \ldots & \K_{2,1}(w_L,v_K) & \K_{2,2}(w_L,y_1) & \ldots & \K_{2,2}(w_L,y_J) \\
\Q_{n+I-K}(v_1) & \ldots & \Q_{n+I-K}(v_K) & q_{n+I-K}(y_1) & \ldots & q_{n+I-K}(y_J) \\
\vdots & & \vdots & \vdots & & \vdots \\
\Q_{n+J-L-1}(v_1) & \ldots & \Q_{n+J-L-1}(v_K) & q_{n+J-L-1}(y_1) & \ldots &
q_{n+J-L-1}(y_J)
\end{pmatrix},
\end{multline}
with $$ C = \left\{\begin{array}{ll} (-1)^{(I+K)L}\prod_{i=0}^{I-K-1}
h^2_{n+i},&\quad\textrm{if $I-K\geq 0$},\\
\vphantom{\frac{1^1_1}{1^1_1}}(-1)^{(I+K)L}\prod_{i=I-K}^{-1}
h^{-2}_{n+i},&\quad\textrm{if $I-K< 0$}.
\end{array}\right.$$
Here $\K_{i,j}$ are as in \eqref{def:K11sum}--\eqref{def:K22sum} but with $n$
replaced by $n+p$ for any arbitrary but fixed number $p\in\{I-K,\ldots,J-L\}$.
\end{theorem}

Note that it would be more natural to write the matrix in \eqref{gencase:a} in
its transposed form; we have written it in its present form because of
typographical reasons.

Theorem~\ref{theorem:avpol:gen} will be proved in
Section~\ref{subsection:proof:gen}. The proof will use
Theorem~\ref{theorem:avpol:num} together with a mechanism to transform an
external source in the matrix $M_1$ (or $M_2$) in the \emph{numerator}, into an
external source in the matrix $M_2$ (or $M_1$ respectively) in the
\emph{denominator}.

In \eqref{gencase:a} we assume that the kernels $\K_{i,j}$ are defined in
\eqref{def:K11sum}--\eqref{def:K22sum} but with $n$ replaced by $n+p$ for an
arbitrary but fixed number $p\in\{J-L,\ldots,I-K\}$. The fact that all values
of $p\in\{J-L,\ldots,I-K\}$ give the same determinant can be seen by applying
elementary row operations to \eqref{gencase:a}. More precisely, to the $j$th
row of \eqref{gencase:a}, $j\in\{1,\ldots,J+K\}$, we may add a suitable linear
combination of rows $J+K+1,\ldots,I+L$, taking into account
\eqref{def:K11sum}--\eqref{def:K22sum}. A similar remark applies to
\eqref{gencase:b}.

In the special case where $J=L=0$ in \eqref{avcharpol}, i.e., when there are
external sources in the matrix $M_1$ only, then Theorem~\ref{theorem:avpol:gen}
could be obtained from the Riemann-Hilbert characterization in
Section~\ref{subsection:proof:num} together with the results in \cite{Del}; see
also Remark~\ref{remark:RHkernels} in Section~\ref{subsection:proof:num}. A
similar remark holds of course if $I=K=0$, i.e., when there are external
sources in the matrix $M_2$ only. We will not use this approach however in the
proofs.
\smallskip

In Theorem~\ref{theorem:avpol:gen} we assume that
\begin{equation}\label{assumption:mild}\min(I-K,J-L)\geq -n.\end{equation}
We believe that this assumption is sufficiently mild to cover most of the
applications. We also expect that similar determinantal formulae may exist if
the assumption \eqref{assumption:mild} fails. See \cite{Bergere2} for the
statement of such formulae in a similar context, requiring a lot of case
distinctions to be made. We will not address this issue in this paper.

We wish to point out that the above formulas have a lot of similarity with
those obtained in the physical papers \cite{Bergere,Bergere2} by Berg\`ere in
the context of the \emph{normal $1$-matrix model} (or more general complex
matrix models), see also \cite{AP,AV}. Recall that the normal $1$-matrix model
is defined by the probability distribution
\begin{equation}\label{normalmatrixmodel}\frac{1}{Z_n} e^{-\textrm{Tr} (V(M,M^*))}\ dM,
\end{equation}
on the space $\mathcal N_n$ of all normal $n\times n$ matrices $M$. Here the
superscript ${}^*$ denotes the Hermitian conjugate, and the potential function
$V$ in \eqref{normalmatrixmodel} is such that its trace is \emph{real-valued}
and depends only on the eigenvalues of $M$. 
In this context, one has the following analogue of \eqref{inteigs}:
\begin{multline}\label{normal:inteigs}
\frac{1}{Z_n}\int_{\mathcal N_n}\! f(M)g(M^*) e^{-\textrm{Tr} (V(M,M^*))}\ dM
\\ = \frac{1}{\til Z_n}\int_{-\infty}^{\infty}\!\ldots\int_{-\infty}^{\infty}
\til f(\lam_1,\ldots,\lam_n)\til
g(\overline{\lam_1},\ldots,\overline{\lam_n})\prod_{i=1}^n \left(
e^{-V(\lam_i,\overline{\lam_i})}
\right)\Delta(\mathbf{\lam})\Delta(\mathbf{\overline{\lam}}) \prod_{i=1}^n
d\lam_i,
\end{multline}
for functions $f$ and $g$ which depend only on the eigenvalues and which are
such that the integrals converge. Here the bar denotes complex conjugation.

In the physical papers \cite{Bergere,Bergere2}, Berg\`ere obtains determinantal
formulae for averages of products and ratios of characteristic polynomials over
the normal matrix model \eqref{normalmatrixmodel} (and also more general
complex matrix models). His formulas are written as determinants constructed
from (bi)orthogonal polynomials and kernels in a similar way as in our results.
We should stress however that the two-matrix model \eqref{twomatrixmodel} and
the normal $1$-matrix model \eqref{normalmatrixmodel} are truly different
models, exhibiting very different properties. Viewed from this perspective, the
similarity of our results with those of Berg\`ere \cite{Bergere,Bergere2} may
appear rather remarkable. It would be interesting to obtain a better
understanding on this issue.

\subsection{Organization of this paper}

The rest of this paper is organized as follows. In Section~\ref{section:proofs}
we prove the main theorems. In Section~\ref{section:appl} we give two
applications: a proof of the Eynard-Mehta theorem for correlation functions,
and a formula for averages of products of traces. Finally,
Section~\ref{section:chain} contains some concluding remarks on the more
general model of matrices coupled in a chain.

\section{Proofs}
\label{section:proofs}

In this section we prove the main theorems. This section is organized as
follows. In Section~\ref{subsection:proof:lemma:sum} we prove
Lemma~\ref{lemma:sumformulas}. Section~\ref{subsection:proof:M1} contains the
proof of Theorems~\ref{theorem:avpol1} and \ref{theorem:avpol2}. In
Section~\ref{subsection:proof:num} we prove Theorem~\ref{theorem:avpol:num}. In
Section~\ref{subsection:reproducing} we discuss several properties of the
kernels $\K_{i,j}$. In Section~\ref{subsection:proof:gen} we use these
properties to prove the general result in Theorem~\ref{theorem:avpol:gen}.

\subsection{Proof of Lemma~\ref{lemma:sumformulas}}
\label{subsection:proof:lemma:sum}

Equations \eqref{def:K12sum} and \eqref{def:K21sum} are immediate by the
definitions. Now let us prove \eqref{def:K11sum}. Using the decomposition
$x-\xi = (x-v)+(v-\xi)$ in the integrand in \eqref{def:K11til}, we obtain
\begin{eqnarray}\nonumber
\K_{1,1}(x,v) &=& \int_{-\infty}^{\infty} \frac{1}{v-\xi}K_{1,1}(x,\xi)\ d\xi
-\frac{1}{v-x}\int_{-\infty}^{\infty} K_{1,1}(x,\xi)\
d\xi\\
&=& \label{def:K11til:bis} \int_{-\infty}^{\infty}
\frac{1}{v-\xi}K_{1,1}(x,\xi)\ d\xi - \frac{1}{v-x},
\end{eqnarray}
where we used that \begin{eqnarray} \nonumber\int_{-\infty}^{\infty}
K_{1,1}(x,\xi)\ d\xi &:=& \int_{-\infty}^{\infty} \sum_{i=0}^{n-1}
\frac{1}{h_i^2}p_i(x)Q_i(\xi)\ d\xi \\ &=& \nonumber
\int_{-\infty}^{\infty}\!\int_{-\infty}^{\infty} \sum_{i=0}^{n-1}
\frac{1}{h_i^2}p_i(x)q_i(\eta)e^{-V(\xi)-W(\eta)+\tau\xi\eta}\ d\xi\ d\eta \\
&=& \nonumber\sum_{i=0}^{n-1} p_i(x)
\left(\int_{-\infty}^{\infty}\!\int_{-\infty}^{\infty} \frac{1}{h_i^2} p_0(\xi)q_i(\eta)e^{-V(\xi)-W(\eta)+\tau\xi\eta}\ d\xi\ d\eta\right) \\
&=& \label{proof:repr1}1,
\end{eqnarray}
due to the biorthogonality relations and the fact that $p_0(\xi)\equiv 1$. Now
by inserting \eqref{def:K11} in \eqref{def:K11til:bis}, we find that
$$ \K_{1,1}(x,v) = \sum_{i=0}^{n-1} \frac{1}{h_i^2}p_i(x)\left(\int_{-\infty}^{\infty}
\frac{1}{v-\xi}Q_i(\xi)\ d\xi\right)-\frac{1}{v-x},
$$
which is \eqref{def:K11sum}.  The proof of \eqref{def:K22sum} is similar, this
time using the following analogue of \eqref{def:K11til:bis}:
\begin{equation} \label{def:K22til:bis} \K_{2,2}(w,y) =
\int_{-\infty}^{\infty}\frac{1}{w-\eta}K_{2,2}(\eta,y)\ d\eta-\frac{1}{w-y}.
\end{equation}
$\bol$

\begin{remark} (Reproducing properties of $K_{1,1}$ and $K_{2,2}$:) The argument
used to obtain \eqref{proof:repr1} in the above proof shows more generally that
\begin{equation}\label{K11:repr} \int_{-\infty}^{\infty} p(\xi) K_{1,1}(x,\xi)\ d\xi = p(x)
\end{equation}
and similarly \begin{equation}\label{K22:repr} \int_{-\infty}^{\infty} q(\eta)
K_{2,2}(\eta,y)\ d\eta = q(y)
\end{equation}
whenever $p$ (or $q$) is equal to one of the biorthogonal polynomials $p_i$ (or
$q_i$ respectively) with $i=0,\ldots,n-1$. By linearity,
\eqref{K11:repr}--\eqref{K22:repr} then remain true for any polynomials $p$ and
$q$ of degree at most $n-1$.

We also remark that in addition to \eqref{K11:repr}--\eqref{K22:repr}, there
are many other formulas of reproducing type for the kernels $K_{i,j}$,
$i,j=1,2$. We prefer to state them not all separately but will sometimes
encounter them in the proofs below.
\end{remark}

\subsection{Proofs of Theorems~\ref{theorem:avpol1} and \ref{theorem:avpol2}}
\label{subsection:proof:M1}

In this section we prove Theorems~\ref{theorem:avpol1} and
\ref{theorem:avpol2}. The proofs will use the Riemann-Hilbert characterization
of Kuijlaars-McLaughlin \cite{KMcL}.

We recall our assumption that the potential $W$ is a polynomial of even degree:
\begin{equation}\label{def:q} W(y) = c_q y^q+O(y^{q-2}),\qquad c_q>0,\qquad q:=\deg
W\in 2\mathbb N.\end{equation} Following \cite{KMcL}, we define the weight
functions
\begin{equation}\label{def:ws} w_j(x) := \int_{-\infty}^{\infty} y^j
e^{-V(x)-W(y)+\tau xy}\ dy,\qquad  j=0,\ldots,q-2.
\end{equation}
We stack these weight functions together in the row vector
\begin{equation}\label{weightmatrixrankone}
\mathbf w(x) := \begin{pmatrix}w_{0}(x) & \ldots & w_{q-2}(x)\end{pmatrix}.
\end{equation}

The Riemann-Hilbert problem (RH problem) of Kuijlaars-McLaughlin \cite{KMcL} is
as follows. See also \cite{BEH2,BHI,EMcL,Kap} for alternative Riemann-Hilbert
problems for the two-matrix model.

\begin{rhp}\label{RHP:Y}
We look for a matrix-valued function $Y(z)$ of size $q$ by $q$ such that
\begin{itemize}
\item[(1)] $Y(z)$ is analytic
in $\mathbb C \setminus\mathbb R $;
 \item[(2)] For $x\in\mathbb R $, it holds that
\begin{equation}\label{jumps:Y}
Y_{+}(x) = Y_{-}(x)
\begin{pmatrix} 1 & \mathbf w(x)\\
0 & I_{q-1}
\end{pmatrix},
\end{equation}
where $I_{q-1}$ denotes the identity matrix of size $q-1$, where $\mathbf w(x)$
is defined in \eqref{weightmatrixrankone}, and where the notation $Y_+(x),
Y_-(x)$ denotes the limit of $Y(z)$ with $z$ approaching $x\in\mathbb R $ from
the upper or lower half plane in $\mathbb C $, respectively;
  \item[(3)] As $z\to\infty$, we have that
\begin{equation}\label{asymptoticcondition:Y} Y(z) =
    (I_{q}+O(1/z))\diag(z^{n},z^{-n_0},\ldots,z^{-n_{q-2}}),
\end{equation}
where $n_k := \lfloor \frac{n+q-2-k}{q-1} \rfloor$, $k=0,\ldots,q-2$.
\end{itemize}
\end{rhp}

Similarly, we may consider the following dual RH problem.

\begin{rhp}\label{RHP:X}
We look for a matrix-valued function $X(z)$ of size $q$ by $q$ such that
\begin{itemize}
\item[(1)] $X(z)$ is analytic
in $\mathbb C \setminus\mathbb R $;
 \item[(2)] For $x\in\mathbb R $, it holds that
\begin{equation}\label{jumps:X}
X_{+}(x) = X_{-}(x)
\begin{pmatrix} I_{q-1} & \mathbf w^T(x)\\
0 & 1
\end{pmatrix},
\end{equation}
where the superscript ${}^T$ denotes the transpose;
  \item[(3)] As $z\to\infty$, we have that
\begin{equation}\label{asymptoticcondition:X} X(z) =
    (I_{q}+O(1/z))\diag(z^{n_0},\ldots,z^{n_{q-2}},z^{-n}).
\end{equation}
\end{itemize}
\end{rhp}

Partition the matrix $Y=Y(z)$ as
\begin{equation}\label{partition:Y} Y(z) =
\begin{array}{ll} & \begin{array}{ll} \ \ \ \ 1 & \ \ \ \ \ q-1 \end{array}\\
\begin{array}{c} 1 \\ q-1\end{array} \hspace{-4mm} & \begin{pmatrix}Y_{1,1}(z) & Y_{1,2}(z)\\ Y_{2,1}(z) &
Y_{2,2}(z)
\end{pmatrix},\end{array}
\end{equation}
where the partition is such that $Y_{1,1}$ has size $1\times 1$. Similarly,
partition the matrix $X=X(z)$ as
\begin{equation}\label{partition:X} X(z) =
\begin{array}{ll} & \begin{array}{ll} \ \ q-1 & \ \ \ \ \ \ \ 1 \end{array}\\
\begin{array}{c} q-1 \\ 1\end{array} \hspace{-4mm} & \begin{pmatrix}X_{1,1}(z) & X_{1,2}(z)\\ X_{2,1}(z) &
X_{2,2}(z)
\end{pmatrix},\end{array}
\end{equation}
%
where the partition is such that $X_{1,1}$ has size $(q-1)\times (q-1)$.

The solutions to the RH problems \ref{RHP:Y} and \ref{RHP:X} are related
through the formula
\begin{equation}\label{dual:partitioning} X = \begin{pmatrix}X_{1,1} & X_{1,2}\\ X_{2,1} &
X_{2,2}
\end{pmatrix} =
\begin{pmatrix} Y_{2,2} & -Y_{2,1} \\ -Y_{1,2} & Y_{1,1}
\end{pmatrix}^{-T},
\end{equation}
where the superscript ${}^{-T}$ denotes the inverse transpose. This was first
shown by Van Assche et al.~in \cite{VAGK}, see also \cite{AvMV,DK1}.

\begin{lemma}\label{lemma:RHchar}
With the partitions \eqref{partition:Y} and \eqref{partition:X} as described
above, we have that
\begin{equation}\label{RH:Y11} Y_{1,1}(z) = p_n(z),\end{equation} which is the
biorthogonal polynomial in \eqref{def:pnqn}, and
\begin{equation}\label{RH:X22} X_{2,2}(z) = \frac{1}{h^2_{n-1}}\Q_{n-1}(z), \end{equation}
where $h_{n-1}$ and $\Q_{n-1}$ are as defined in \eqref{def:hk} and
\eqref{def:Qn:til} respectively.
\end{lemma}

\begin{proof} The statement \eqref{RH:Y11} was established by Kuijlaars-McLaughlin
\cite{KMcL}. 

Next we prove \eqref{RH:X22}. To this end we recall some of the identities of
Van Assche et al. \cite{VAGK}: it is observed in the latter paper that
\begin{equation}\label{VA1}
X_{2,1}(z) = \begin{pmatrix} A_{0}(z) & \ldots & A_{q-2}(z)
\end{pmatrix}
\end{equation}
for certain polynomials $A_j(z)$ of degree at most $n_j-1$, $j=0,\ldots,q-2$.
These polynomials are such that the following \lq type~1\rq\ orthogonality
relations hold: \begin{equation}\label{VA2} \int_{-\infty}^{\infty} x^j
X_{2,1}(x)\mathbf{w}^T(x)\ dx = -2\pi i\delta_{j,n-1},
\end{equation}
for $j=0,1,\ldots,n-1$, where $\delta_{j,n-1}$ denotes the Kronecker delta.
Finally, \cite{VAGK} also shows that
\begin{equation}\label{VA3} X_{2,2}(z) = -\frac{1}{2\pi
i}\int_{-\infty}^{\infty} \frac{1}{z-x} X_{2,1}(x)\mathbf{w}^T(x)\ dx.
\end{equation}

Now we are going to use the special form of the weight functions $w_j(x)$ in
\eqref{def:ws}--\eqref{weightmatrixrankone}. Using \eqref{VA1} and an
integration by parts argument as in \cite{KMcL} it is easy to see that
\begin{eqnarray}\nonumber X_{2,1}(x)\mathbf{w}^T(x) &:=& A_0(x)w_0(x)+\ldots+A_{q-2}(x)w_{q-2}(x) \\
\label{VA4}  &=& \int_{-\infty}^{\infty} B_{n-1}(y) e^{-V(x)-W(y)+\tau xy}\ dy,
\end{eqnarray}
for a certain polynomial $B_{n-1}(y)$ of degree at most $n-1$. Inserting this
in \eqref{VA2}--\eqref{VA3}, these relations are transformed into
\begin{equation}\label{VA5} \int_{-\infty}^{\infty}\!\int_{-\infty}^{\infty}\!
x^j B_{n-1}(y) e^{-V(x)-W(y)+\tau xy}\ dx\ dy = -2\pi i\delta_{j,n-1},
\end{equation}
for $j=0,1,\ldots,n-1$, and \begin{equation}\label{VA6}  X_{2,2}(z) =
-\frac{1}{2\pi i} \int_{-\infty}^{\infty}\!\int_{-\infty}^{\infty}\!
\frac{1}{z-x} B_{n-1}(y) e^{-V(x)-W(y)+\tau xy}  \ dx\ dy,
\end{equation}
respectively.

Now \eqref{VA5} shows that $B_{n-1}$ is precisely the biorthogonal polynomial
$q_{n-1}$ in \eqref{def:pnqn}, up to a constant, i.e.,
$$B_{n-1}(y) = C_{n-1} q_{n-1}(y),$$ for certain $C_{n-1}\in\cee$. Then
\eqref{VA6} implies in turn that \begin{eqnarray} \label{VA7} X_{2,2}(z) &=&
-\frac{C_{n-1}}{2\pi i} \int_{-\infty}^{\infty}\!\int_{-\infty}^{\infty}\!
\frac{1}{z-x} q_{n-1}(y) e^{-V(x)-W(y)+\tau xy}  \ dx\ dy\\ \label{VA8} &=&
-\frac{C_{n-1}}{2\pi i} \Q_{n-1}(z),
\end{eqnarray}
where the second equality follows by virtue of \eqref{def:Qn:til}.

Finally, we want to identify the normalization constant $C_{n-1}$ in
\eqref{VA8}. To this end, we substitute the expansion
$\frac{1}{z-x}=\frac{1}{z}+\frac{x}{z^2}+\frac{x^2}{z^3}+\ldots$ and the
orthogonality relations \eqref{def:pnqn} in \eqref{VA7} to obtain that
$$ X_{2,2}(z) = -\frac{C_{n-1}}{2\pi i} h^2_{n-1}z^{-n}+O(z^{-n-1}),\qquad z\to\infty.
$$
Comparing this with the asymptotics in \eqref{asymptoticcondition:X}, we see
that
$$ -\frac{C_{n-1}}{2\pi i}h^2_{n-1} = 1.$$ Inserting this expression for $C_{n-1}$ in
\eqref{VA8}, we obtain the desired formula \eqref{RH:X22}. This ends the proof
of the lemma.
\end{proof}

\textit{Proof of Theorems \ref{theorem:avpol1} and \ref{theorem:avpol2}.}
Taking into account Lemma~\ref{lemma:RHchar}, Theorems \ref{theorem:avpol1} and
\ref{theorem:avpol2} now follow directly from the results in \cite{BK1} and
\cite{Desrosiers1} respectively; see also \cite{Del,Kui}.$\bol$

\begin{remark}\label{remark:RHkernels} The kernels $K_{1,1}$ in \eqref{def:K11} and $\K_{1,1}$ in
\eqref{def:K11til:bis} can be expressed in terms of the RH problem for $Y(z)$
as follows:
\begin{equation}\label{RH:kernel} K_{1,1}(x_1,x_2) = \frac{1}{2\pi i(x_1-x_2)}
\begin{pmatrix} 0 & \mathbf{w}(x_2) \end{pmatrix}
Y^{-1}(x_2)Y(x_1)
\begin{pmatrix}1 \\ 0\\ \vdots \\ 0 \end{pmatrix},
\end{equation}
and
\begin{equation}\label{RH:kernel:til} \K_{1,1}(x,v) = \frac{1}{x-v}
\begin{pmatrix} 1 & 0 & \ldots & 0 \end{pmatrix}
Y^{-1}(v)Y(x)
\begin{pmatrix}1 \\ 0\\ \vdots \\ 0 \end{pmatrix}.
\end{equation}
These formulas can be obtained from the results in \cite{DK1} and \cite{Del}
respectively. Note that if $x_1,x_2\in\er$ (or $x\in\er$) then we should
replace $Y$ in \eqref{RH:kernel} (or \eqref{RH:kernel:til} respectively) by one
of its boundary values $Y_+$ or $Y_-$; both boundary values lead to the same
formulas. Also note that, by means of \eqref{RH:kernel:til} (or
\eqref{def:K11til:bis}) and the results in \cite{Del}, it is possible to obtain
determinantal formulas for $P_n^{[I,0,K,0]}$ for any $I$ and $K$; however we
will not follow this route here.
\end{remark}

\subsection{Proof of Theorem~\ref{theorem:avpol:num}}
\label{subsection:proof:num}

In this section we prove Theorem~\ref{theorem:avpol:num}. The proof will be
virtually the same as the one of Akemann et al.~\cite{ADOS}; we include it for
convenience of the reader. The proof will require two formulas of Christoffel
type, see \cite{ADOS,BDS,Sz}.

\begin{proposition}\label{prop:Chris1} Define the weight function
\begin{equation}\label{def:wtil1}\til
w(x,y):=\left(\prod_{i=1}^{I}(x-x_i)\right)e^{-V(x)-W(y)+\tau
xy},\end{equation} with $I\geq 0$. The monic biorthogonal polynomial $A_n(x)$
of degree $n$ with respect to this weight function, which is defined by the
biorthogonality relations
\begin{equation}\label{Chris:bior1}
\int_{-\infty}^{\infty}\!\int_{-\infty}^{\infty}\! A_n(x)q(y)\til w(x,y)\ dx\
dy = 0
\end{equation}
for any polynomial $q(y)$ of degree at most $n-1$, is given by
\begin{equation}\label{def:An}
A_n(x):=\frac{1}{\prod_{i=1}^{I}(x-x_i)} \frac{\det\begin{pmatrix}
p_{n}(x_{1}) & \ldots & p_{n+I}(x_{1})\\
\vdots & & \vdots\\
p_{n}(x_{I}) & \ldots &
p_{n+I}(x_{I})\\
p_{n}(x) & \ldots & p_{n+I}(x)
\end{pmatrix}} {\det\begin{pmatrix}
p_{n}(x_{1}) & \ldots & p_{n+I-1}(x_{1})\\
\vdots & & \vdots\\
p_{n}(x_{I}) & \ldots & p_{n+I-1}(x_{I})\end{pmatrix}}.
\end{equation}
\end{proposition}

\begin{proof}
Clearly $A_n(x)$ is a monic polynomial in $x$ of degree $n$. So it suffices to
check the biorthogonality relations \eqref{Chris:bior1}. Observe that in the
integrand $A_n(x)q(y)\til w(x,y)$, the factor $\prod_{i=1}^{I}(x-x_i)$ of $\til
w(x,y)$ in \eqref{def:wtil1} cancels with the prefactor of $A_n(x)$ in
\eqref{def:An}. By the linearity of determinants, the double integral in
\eqref{Chris:bior1} can be performed entrywise inside the last row of the
matrix in the numerator of \eqref{def:An}. Then it suffices to prove that
$$ \int_{-\infty}^{\infty}\!\int_{-\infty}^{\infty}\!
p_{n+i}(x)q(y)w(x,y)\ dx\ dy = 0,\qquad i=0,\ldots,I,
$$
which is a direct consequence of the biorthogonality relations.
\end{proof}

\begin{proposition}\label{prop:Chris2} Define the weight function
\begin{equation}\label{def:wtil2}\til
w(x,y):=\left(\prod_{i=1}^{I}(x-x_i)\prod_{j=1}^{J}(y-y_j)\right)e^{-V(x)-W(y)+\tau
xy},\end{equation} where $I>J\geq 0$. The monic biorthogonal polynomial
$B_n(y)$ of degree $n$ with respect to this weight function, which is defined
by the biorthogonality relations \begin{equation}\label{Chris:bior3}
\int_{-\infty}^{\infty}\!\int_{-\infty}^{\infty}\! p(x)B_n(y)\til w(x,y)\ dx\
dy = 0
\end{equation} for any polynomial $p$ of degree at most $n-1$, is given by
\begin{multline}\label{def:Bn}
B_n(y):=\frac{h_{n+J}^2}{\prod_{j=1}^{J}(y-y_j)}\\ \times
\frac{\det\begin{pmatrix}
\K_{1,2}(x_1,y_1) & \ldots & \K_{1,2}(x_1,y_J) & \K_{1,2}(x_1,y) & p_{n+J+1}(x_{1}) & \ldots & p_{n+I-1}(x_{1})\\
\vdots & & \vdots & \vdots & & \vdots\\
\K_{1,2}(x_I,y_1) & \ldots & \K_{1,2}(x_I,y_J)& \K_{1,2}(x_I,y) &
p_{n+J+1}(x_{I}) & \ldots &
p_{n+I-1}(x_{I})\\
\end{pmatrix}} {\det\begin{pmatrix}
\K_{1,2}(x_1,y_1) & \ldots & \K_{1,2}(x_1,y_J) & p_{n+J}(x_{1}) & p_{n+J}(x_{1}) & \ldots & p_{n+I-1}(x_{1})\\
\vdots & & \vdots & \vdots & & \vdots\\
\K_{1,2}(x_I,y_1) & \ldots & \K_{1,2}(x_I,y_J)& p_{n+J}(x_{1}) & p_{n+J}(x_{I})
& \ldots & p_{n+I-1}(x_{I})\end{pmatrix}}.
\end{multline}
Here we define $\K_{1,2}(x,y)$ as in \eqref{def:K12sum} but with $n$ replaced
by $n+p$ for an arbitrary but fixed number $p\in\{J+1,\ldots,I\}$. 
\end{proposition}

\begin{proof} Clearly $B_n(y)$ is a monic polynomial in $y$ of degree $n$. Next
we check the biorthogonality relations \eqref{Chris:bior3}. To this end we will
take the index $p$ in the statement of the lemma equal to $I$. Observe that in
the integrand $p(x)B_n(y)\til w(x,y)$, the factor $\prod_{j=1}^{J}(y-y_j)$ of
$\til w(x,y)$ in \eqref{def:wtil2} cancels with the prefactor of $B_n(y)$ in
\eqref{def:Bn}. By linearity we can then take the double integral in
\eqref{Chris:bior3} entrywise inside the $(J+1)$th column of the matrix in the
numerator of \eqref{def:Bn}. Hence the integral \eqref{Chris:bior3} can be
written as a linear combination of terms of the form
\begin{eqnarray}\nonumber & &
\int_{-\infty}^{\infty}\!\int_{-\infty}^{\infty}\!
p(x)\prod_{i=1}^{I}(x-x_i)\K_{1,2}(x_k,y) e^{-V(x)-W(y)+\tau xy}\ dx\ dy
 \\ \label{Chris:bior4} &=& \int_{-\infty}^{\infty}\!
p(x)\prod_{i=1}^{I}(x-x_i)K_{1,1}(x_k,x)\ dx,
\end{eqnarray}
for $k=1,\ldots,I$, where we used \eqref{def:K11} and \eqref{def:K12sum}. But
the integral \eqref{Chris:bior4} is zero for any $k=1,\ldots,I$, because of the
reproducing property \eqref{K11:repr} for $K_{1,1}$. (Recall that we are taking
the index $p$ in the statement of the lemma equal to $I$.)
\end{proof}

Theorem~\ref{theorem:avpol:num} now follows from Theorem~\ref{theorem:avpol1}
and Propositions~\ref{prop:Chris1} and \ref{prop:Chris2} by an easy induction
argument, see e.g.\ \cite{BDS}.$\bol$

\subsection{Properties of the kernels $\K_{i,j}$}
\label{subsection:reproducing}

Our next goal is to prove Theorem~\ref{theorem:avpol:gen}. In the present
section we first collect some preliminary results on the kernels $\K_{i,j}$
which will be needed in the proof.

First we discuss the kernel $\K_{1,2}(x,y) = K_{1,2}(x,y)$. Observe that by the
definition \eqref{def:K12sum}, the kernel $\K_{1,2}(x,y)$ is a bivariate
polynomial in $x$ and $y$, of the form
\begin{equation}\label{K12:coeff} \K_{1,2}(x,y) = \sum_{i,j=0}^{n-1}
c_{i,j}x^iy^j,
\end{equation}
for suitable coefficients $c_{i,j}\in\cee$.

\begin{lemma}\label{lemma:K12repr} (Reproducing property of $\K_{1,2}$:)
The kernel $\K_{1,2}$ is reproducing in the sense that
\begin{equation}\label{K12:repr:a}
\int_{-\infty}^{\infty}\!\int_{-\infty}^{\infty}\!
\K_{1,2}(x,y)p(\xi)e^{-V(\xi)-W(y)+\tau \xi y}\ d\xi\ dy = p(x),
\end{equation}
\begin{equation}\label{K12:repr:b}
\int_{-\infty}^{\infty}\!\int_{-\infty}^{\infty}\!
\K_{1,2}(x,y)q(\eta)e^{-V(x)-W(\eta)+\tau x\eta}\ dx\ d\eta = q(y),
\end{equation}
for any polynomials $p$ and $q$ of degree at most $n-1$. Moreover, the
reproducing property \eqref{K12:repr:a} (or \eqref{K12:repr:b}) uniquely
characterizes $\K_{1,2}(x,y)$ over all bivariate polynomials of the form
\eqref{K12:coeff}.

\end{lemma}

\begin{proof} Taking into account \eqref{def:K11}, \eqref{def:K22} and \eqref{def:K12sum},
the reproducing properties \eqref{K12:repr:a}--\eqref{K12:repr:b} are just a
restatement of \eqref{K11:repr} and \eqref{K22:repr} respectively. The claim
about uniqueness is also easily shown, see also \cite{Del}.
\end{proof}

Next we discuss the kernels $\K_{1,1}$ and $\K_{2,2}$. We will need the
following result.

\begin{lemma}\label{lemma:vanishing} (Vanishing properties of $\K_{1,1}$ and $\K_{2,2}$:)
The kernel $\K_{1,1}$ satisfies the \lq vanishing property\rq\
\begin{equation}\label{vanishing:K11}
\int_{-\infty}^{\infty}\!\int_{-\infty}^{\infty}\!
\K_{1,1}(\xi,v)q(\eta)e^{-V(\xi)-W(\eta)+\tau \xi\eta} \ d\xi\ d\eta = 0
\end{equation}
for any $v\in\cee\setminus\er$ and all polynomials $q$ of degree at most $n-1$.
Similarly, the kernel $\K_{2,2}$ satisfies the vanishing property
\begin{equation}\label{vanishing:K22}
\int_{-\infty}^{\infty}\!\int_{-\infty}^{\infty}\!
\K_{2,2}(w,\eta)p(\xi)e^{-V(\xi)-W(\eta)+\tau \xi \eta} \ d\xi\ d\eta = 0,
\end{equation}
for any $w\in\cee\setminus\er$ and all polynomials $p$ of degree at most $n-1$.
\end{lemma}

\begin{proof} This lemma
can be obtained from \eqref{RH:kernel:til} and the vanishing property in
\cite{Del}. For completeness, we include a direct proof. Using
\eqref{def:K11til:bis} we have
\begin{multline}\label{proof:vanishing}
\int_{-\infty}^{\infty}\!\int_{-\infty}^{\infty}\!
\K_{1,1}(\xi,v)q(\eta)e^{-V(\xi)-W(\eta)+\tau \xi\eta} \ d\xi\ d\eta =
-\int_{-\infty}^{\infty}\!\int_{-\infty}^{\infty}\!
\frac{q(\eta)e^{-V(\xi)-W(\eta)+\tau \xi\eta}}{v-\xi} \ d\xi\ d\eta \\ +
\int_{-\infty}^{\infty}\!\int_{-\infty}^{\infty}\!\int_{-\infty}^{\infty}\!
\frac{K_{1,1}(\xi,x)q(\eta)e^{-V(\xi)-W(\eta)+\tau \xi\eta}}{v-x} \ d\xi\
d\eta\ dx.
\end{multline}
Now we use the following formula of reproducing type:
\begin{equation}\label{K11:repr:bis}
\int_{-\infty}^{\infty}\!\int_{-\infty}^{\infty}\!
K_{1,1}(\xi,x)q(\eta)e^{-V(\xi)-W(\eta)+\tau \xi\eta} \ d\xi\ d\eta =
\int_{-\infty}^{\infty}\! q(\eta)e^{-V(x)-W(\eta)+\tau x\eta}\ d\eta,
\end{equation}
for any polynomial $q$ of degree at most $n-1$. (As usual, by linearity it
suffices to check \eqref{K11:repr:bis} when $q$ equals one of the biorthogonal
polynomials $q_i$ for $i=0,\ldots,n-1$, in which case it follows easily by the
biorthogonality relations.) Using \eqref{K11:repr:bis} in the term with the
triple integral in \eqref{proof:vanishing}, we see that both terms in the right
hand side of \eqref{proof:vanishing} cancel each other and so
\eqref{vanishing:K11} follows. The proof of \eqref{vanishing:K22} is similar.
\end{proof}

\begin{lemma}\label{lemma:leadingterms} (Asymptotic behavior of  biorthogonal polynomials and kernels:)
We have
\begin{eqnarray}
\label{p:asy} p_n(x) &=& x^{n}+O(x^{n-1}),\qquad x\to\infty,\\
\label{q:asy} q_n(x) &=& y^{n}+O(y^{n-1}),\qquad y\to\infty,\\
\label{Ptil:asy} \P_n(w) &=& h_n^2 w^{-n-1}+O(w^{-n-2}),\qquad w\to\infty, \\
\label{Qtil:asy} \Q_n(v) &=& h_n^2 v^{-n-1}+O(v^{-n-2}),\qquad v\to\infty,
\end{eqnarray}
and \begin{eqnarray}\label{K11til:asy1}\K_{1,1}(x,v) &=&
h_{n-1}^{-2}\Q_{n-1}(v)x^{n-1}+O(x^{n-2}), \qquad x\to\infty,\\
\label{K11til:asy2}\K_{1,1}(x,v) &=& -p_n(x)v^{-n-1}+O\left(v^{-n-2}\right),
\qquad v\to\infty,
\\
\label{K12til:asy1}\K_{1,2}(x,y) &=&
h_{n-1}^{-2}q_{n-1}(y)x^{n-1}+O(x^{n-2}),\qquad x\to\infty,\\
\label{K12til:asy2}\K_{1,2}(x,y) &=&
h_{n-1}^{-2}p_{n-1}(x)y^{n-1}+O(y^{n-2}),\qquad y\to\infty,
\end{eqnarray}
and
\begin{eqnarray}
\label{K21til:asy1}\K_{2,1}(w,v) &=&
-\Q_n(v)w^{-n-1}+O\left(w^{-n-2}\right),\qquad w\to\infty,\\
\label{K21til:asy2}\K_{2,1}(w,v) &=&
-\P_n(w)v^{-n-1}+O\left(v^{-n-2}\right),\qquad v\to\infty,\\
\label{K22til:asy1}\K_{2,2}(w,y) &=& -q_n(x)w^{-n-1}+O\left(w^{-n-2}\right),
\qquad w\to\infty,\\
\label{K22til:asy2}\K_{2,2}(w,y)
&=&h_{n-1}^{-2}\P_{n-1}(w)y^{n-1}+O(y^{n-2}),\qquad y\to\infty.
\end{eqnarray}
\end{lemma}

\begin{proof} Equations \eqref{p:asy}--\eqref{q:asy} are obvious by definition.
Let us check \eqref{Ptil:asy}. By substituting the series $\frac{1}{w-\eta} =
\sum_{i=1}^{\infty}\frac{\eta^{i-1}}{w^i}$ in \eqref{def:Pn:til} we find that
$$ \P_n(w) = \sum_{i=1}^{\infty}\frac{1}{w^i}\left(
\int_{-\infty}^{\infty}\!\int_{-\infty}^{\infty}\!
p_n(\xi)\eta^{i-1}e^{-V(\xi)-W(\eta)+\tau\xi\eta}\ d\xi\ d\eta\right).
$$
Now by the biorthogonality relations we have
$$ \int_{-\infty}^{\infty}\!\int_{-\infty}^{\infty}\!
p_n(\xi)\eta^{i-1} e^{-V(\xi)-W(\eta)+\tau\xi\eta}\ d\xi\ d\eta =
h_n^2\delta_{i,n+1},
$$
for any $i=1,\ldots,n+1$. So we obtain \eqref{Ptil:asy}. In a similar way one
checks \eqref{Qtil:asy}.

Equation \eqref{K11til:asy1} follows by virtue of \eqref{def:K11sum} and
\eqref{p:asy}. Next we check \eqref{K11til:asy2}. By substituting the expansion
$\frac{1}{v-\xi} = \sum_{i=1}^{\infty}\frac{\xi^{i-1}}{v^i}$ in
\eqref{def:K11til} we obtain
\begin{equation}\label{proofdom0}\K_{1,1}(x,v)=\frac{1}{x-v}
\sum_{i=1}^{\infty}\frac{1}{v^i}\left(\int_{-\infty}^{\infty}
(x-\xi)\xi^{i-1}K_{1,1}(x,\xi)\ d\xi\right).\end{equation} Now
$$ \int_{-\infty}^{\infty}\! (x-\xi)\xi^{i-1}K_{1,1}(x,\xi)\ d\xi = \left[ (x-\xi)\xi^{i-1}
\right]_{\xi=x} = 0,\qquad \textrm{if } i=1,\ldots,n-1,
$$
because of the reproducing property \eqref{K11:repr}. For $i=n$ we have
\begin{multline}\label{K11:ntonplusone} \int_{-\infty}^{\infty}\! (x-\xi)\xi^{n-1}K_{1,1}(x,\xi)\ d\xi
= \int_{-\infty}^{\infty}\!
(x-\xi)\xi^{n-1}\left(K_{1,1}(x,\xi)+\frac{1}{h_n^2}p_n(x)Q_n(\xi)\right)\ d\xi
\\ -\int_{-\infty}^{\infty}\! (x-\xi)\xi^{n-1}\frac{1}{h_n^2}p_n(x)Q_n(\xi)\
d\xi.
\end{multline}
The expression between brackets in the first term in the right hand side of
\eqref{K11:ntonplusone} is nothing but the kernel $K_{1,1}(x,\xi)$ with $n$
replaced by $n+1$ in \eqref{def:K11}, so this integral is zero as before. We
then have
\begin{eqnarray}\nonumber \int_{-\infty}^{\infty}\! (x-\xi)\xi^{n-1}K_{1,1}(x,\xi)\
d\xi &=& -\frac{1}{h_n^2}p_n(x)\int_{-\infty}^{\infty}\!
(x-\xi)\xi^{n-1}Q_n(\xi)\ d\xi \\
\nonumber &=&
-\frac{1}{h_n^2}p_n(x)\int_{-\infty}^{\infty}\!\int_{-\infty}^{\infty}\!
(x-\xi)\xi^{n-1}q_n(\eta)e^{-V(\xi)-W(\eta)+\tau\xi\eta}\ d\xi\ d\eta\\
\nonumber &=& p_n(x),
\end{eqnarray}
because of biorthogonality. Using this relation in \eqref{proofdom0}, we obtain
\eqref{K11til:asy2}.

Equations \eqref{K12til:asy1}--\eqref{K12til:asy2} are immediate by the
definition \eqref{def:K12sum}. Equations
\eqref{K22til:asy1}--\eqref{K22til:asy2} can be obtained similarly as before.

Finally, let us check \eqref{K21til:asy1}.  By substituting the expansion
$\frac{1}{w-\eta} = \sum_{i=1}^{\infty}\frac{\eta^{i-1}}{w^i}$ in
\eqref{def:K21til} we obtain
\begin{equation}\label{proofdom3}
\K_{2,1}(w,v) =
\sum_{i=1}^{\infty}\frac{1}{w^i}\int_{-\infty}^{\infty}\frac{1}{v-\xi}\left(\int_{-\infty}^{\infty}\eta^{i-1}K_{2,1}(\eta,\xi)\
d\eta\right) d\xi.
\end{equation}
By \eqref{def:K21} we see that \begin{equation}\label{proofdom4}
\int_{-\infty}^{\infty}\! \eta^{i-1}K_{2,1}(\eta,\xi)\ d\eta =
-\int_{-\infty}^{\infty}\! \eta^{i-1}e^{-V(\xi)-W(\eta)+\tau\xi\eta}\ d\eta
+\int_{-\infty}^{\infty}\!
\eta^{i-1}\sum_{j=0}^{n-1}\frac{1}{h_j^2}P_j(\eta)Q_j(\xi)\ d\eta.
\end{equation}
Now we claim that \begin{equation}\label{K22:reprbis} \int_{-\infty}^{\infty}\!
q(\eta)\sum_{j=0}^{n-1}\frac{1}{h_j^2}P_j(\eta)Q_j(\xi)\ d\eta =
\int_{-\infty}^{\infty}\! q(\eta)e^{-V(\xi)-W(\eta)+\tau\xi\eta}\ d\eta
\end{equation}
for any polynomial $q$ of degree at most $n-1$. (Once again, by linearity it
suffices to check \eqref{K22:reprbis} when $q$ equals one of the biorthogonal
polynomials $q_i$ for $i=0,\ldots,n-1$, in which case it follows directly from
the biorthogonality relations.) Using \eqref{K22:reprbis} with
$q(\eta)=\eta^{i-1}$, we see that if $i=1,\ldots,n$ then both terms in the
right hand side of \eqref{proofdom4} cancel each other. For $i=n+1$ we get
\begin{eqnarray*} & & \int_{-\infty}^{\infty}\!
\eta^{n}\sum_{j=0}^{n-1}\frac{1}{h_j^2}P_j(\eta)Q_j(\xi)\ d\eta \\ &=&
\left(\int_{-\infty}^{\infty}\!
\eta^{n}\sum_{j=0}^{n}\frac{1}{h_j^2}P_j(\eta)Q_j(\xi)\
d\eta\right)-\int_{-\infty}^{\infty}\!\eta^n \frac{1}{h_n^2}P_n(\eta)Q_n(\xi)\
d\eta
\\ &=& \left(\int_{-\infty}^{\infty}\! \eta^{n}e^{-V(\xi)-W(\eta)+\tau\xi\eta}\
d\eta\right)-\frac{Q_n(\xi)}{h_n^2}\int_{-\infty}^{\infty}\!\int_{-\infty}^{\infty}\!\eta^n
p_n(x)e^{-V(x)-W(\eta)+\tau x\eta}\ dx\ d\eta
\\ &=& \left(\int_{-\infty}^{\infty}\!
\eta^{n}e^{-V(\xi)-W(\eta)+\tau\xi\eta}\ d\eta\right) -Q_n(\xi),
\end{eqnarray*}
where we used \eqref{K22:reprbis} (with $n$ replaced by $n+1$) and the
biorthogonality relations. Substituting this in \eqref{proofdom4} we find that
$$\int_{-\infty}^{\infty}\! \eta^{i-1}K_{2,1}(\eta,\xi)\ d\eta =
-\delta_{i,n+1}Q_n(\xi),$$ for any $i=1,\ldots,n+1$. Using this in
\eqref{proofdom3} we obtain
$$ \K_{2,1}(w,v)= -w^{-n-1}\int_{-\infty}^{\infty}\!
\frac{Q_n(\xi)}{v-\xi}\ d\xi+O(w^{-n-2}) = -w^{-n-1}\Q_{n}(v)+O(w^{-n-2}),
$$
which is \eqref{K21til:asy1}. Equation \eqref{K21til:asy2} can be obtained
similarly.
\end{proof}

Finally, we will also need the following additional properties of the kernels
$\K_{i,j}$.

\begin{lemma}\label{lemma:kernel:integrate} (Integral formula relations between the kernels $\K_{i,j}$:)
We have
\begin{equation}\label{eqint:P}
\P_n(w) = \int_{-\infty}^{\infty}\!\int_{-\infty}^{\infty}\!
p_n(\xi)e^{-V(\xi)-W(\eta)+\tau\xi\eta}\frac{1}{w-\eta}\ d\xi\ d\eta,
\end{equation}
\begin{equation}\label{eqint:K22}
\K_{2,2}(w,y) = \frac{1}{y-w}\int_{-\infty}^{\infty}\!\int_{-\infty}^{\infty}\!
\K_{1,2}(\xi,y)e^{-V(\xi)-W(\eta)+\tau\xi\eta}\frac{y-\eta}{w-\eta}\ d\xi\
d\eta,
\end{equation}
and
\begin{equation}\label{eqint:K21}
\K_{2,1}(w,v) = \int_{-\infty}^{\infty}\!\int_{-\infty}^{\infty}\!
\K_{1,1}(\xi,v)e^{-V(\xi)-W(\eta)+\tau\xi\eta}\frac{1}{w-\eta}\ d\xi\ d\eta.
\end{equation}
\end{lemma}

\begin{proof}
Equation \eqref{eqint:P} is just \eqref{def:Pn:til}. Equation \eqref{eqint:K22}
follows from the definitions in \eqref{def:K22til}, \eqref{def:K22} and
\eqref{def:K12sum}. Finally, equation \eqref{eqint:K21} follows by
\eqref{def:K21sum}, \eqref{def:K11sum} and \eqref{def:Pn:til}.
\end{proof}

\subsection{Proof of Theorem~\ref{theorem:avpol:gen}}
\label{subsection:proof:gen}

In this section we prove Theorem~\ref{theorem:avpol:gen}. To this end we use
Theorem~\ref{theorem:avpol:num} together with a mechanism to transform an
external source in the matrix $M_1$ (or $M_2$) in the \emph{numerator}, into an
external source in the matrix $M_2$ (or $M_1$ respectively) in the
\emph{denominator}. That is, we will apply transformations of the form
\begin{equation}\label{stepinduc1} (I,J,K,L) \to (I-1,J,K,L+1),
\end{equation}
or \begin{equation}\label{stepinduc2} (I,J,K,L) \to (I,J-1,K+1,L).
\end{equation}
It will then suffice to show that each of the objects in
\eqref{def:Pn:til}--\eqref{def:K22til} transforms in an appropriate way under
these transformations; this will be achieved by virtue of
Lemmas~\ref{lemma:K12repr}--\ref{lemma:kernel:integrate}.

As in \cite{BDS,Del}, we will also need an appropriate use of partial fraction
decomposition. For given complex numbers
$y_1,\ldots,y_J,w_1,\ldots,w_{L+1}\in\cee$, the partial fraction decomposition
which is of interest to us is of the form
\begin{equation}\label{pfd1}\frac{\prod_{j=1}^J (\eta-y_j)}{(\eta-w_{L+1})\prod_{l=1}^L (\eta-w_l)}
= \frac{c_{L+1}}{\eta-w_{L+1}}+\sum_{l=1}^L \frac{c_l}{\eta-w_l}+P(\eta).
\end{equation}
Here $c_1,\ldots,c_{L+1}\in\cee$ and $P$ is a polynomial of degree $J-L-1$. (We
put $P\equiv 0$ when $J-L-1<0$.)

Note that \eqref{pfd1} also implies the more complicated partial fraction
decomposition
\begin{equation}\label{pfd2} \frac{\prod_{j=1}^J (\eta-y_j)}{(\eta-w_{L+1})\prod_{l=1}^L (\eta-w_l)}
= \frac{c_{L+1}(\eta-y_j)}{(w_{L+1}-y_j)(\eta-w_{L+1})}+\sum_{l=1}^L
\frac{c_l(\eta-y_j)}{(w_l-y_j)(\eta-w_l)}+P(\eta)-P(y_j),
\end{equation}
for any fixed $j\in\{1,\ldots,J\}$. We leave it to the reader to obtain
\eqref{pfd2} from \eqref{pfd1}.

Now we are ready for the proof of Theorem~\ref{theorem:avpol:gen}.
\smallskip

\textit{Proof of Theorem~\ref{theorem:avpol:gen}}. By symmetry we can assume
without loss of generality that $I-K\geq J-L$. We will show that the formula in
\eqref{gencase:a} is compatible with any transformation of the form
\eqref{stepinduc1}.

Consider the weight function
\begin{equation}\label{wtilUv}
\til w(x,y) = \frac{\prod_{i=1}^{I-1}(x-x_i)}{\prod_{k=1}^{K}(x-v_k)}
\frac{\prod_{j=1}^{J}(y-y_j)}{\prod_{l=1}^{L}(y-w_l)} e^{-V(x)-W(y)+\tau xy}.
\end{equation}
(Note that there are only $I-1$ factors $x_i$.) Denote by $A_n(x)$ the monic
biorthogonal polynomial of degree $n$ with respect to this weight function,
defined by the orthogonality relations
\begin{equation}
 \int_{-\infty}^{\infty}\!\int_{-\infty}^{\infty}\! A_n(x)q(y)\til
 w(x,y)\ dx\ dy = 0
\end{equation}
for all polynomials $q$ of degree at most $n-1$. Theorem~\ref{theorem:avpol1}
implies that
$$
P_{n}^{[I,J,K,L]}(x_1,\ldots,x_I;y_1,\ldots,y_J;v_1,\ldots,v_K;w_1,\ldots,w_L)
= A_n(x_I),
$$
while from Theorem~\ref{theorem:avpol2} and \eqref{def:Pn:til} it follows that
\begin{multline*}
P_{n}^{[I-1,J,K,L+1]}(x_1,\ldots,x_{I-1};y_1,\ldots,y_J;v_1,\ldots,v_K;w_1,\ldots,w_{L+1})
\\ = \frac{1}{\til h_{n-1}^2}\int_{-\infty}^{\infty}\!\int_{-\infty}^{\infty}\!
A_{n-1}(x_I)\frac{\til w(x_I,\eta)}{w_{L+1}-\eta}\ dx_I\ d\eta, \end{multline*}
where $\til h_{n-1}^2$ is a constant which depends on each of the numbers
$x_i,y_j,v_k,w_l$, but \emph{not} on $w_{L+1}$. By combining these two
formulas, one gets
\begin{multline}\label{mechanism:induction}
P_{n}^{[I-1,J,K,L+1]}(x_1,\ldots,x_{I-1};y_1,\ldots,y_J;v_1,\ldots,v_K;w_1,\ldots,w_{L+1})
\\ = \frac{1}{\til h_{n-1}^2}\int_{-\infty}^{\infty}\!\int_{-\infty}^{\infty}\!
P_{n-1}^{[I,J,K,L]}(x_1,\ldots,x_I;y_1,\ldots,y_J;v_1,\ldots,v_K;w_1,\ldots,w_L)
\frac{\til w(x_I,\eta)}{w_{L+1}-\eta}\ dx_I\ d\eta.
\end{multline}
Now assume, by induction, that \eqref{gencase:a} holds for $P_{n}^{[I,J,K,L]}$.
By substituting this expression (with $n-1$ instead of $n$) for
$P_{n-1}^{[I,J,K,L]}$ in the right hand side of \eqref{mechanism:induction}, we
get
\begin{multline}\label{proof:gen0}
P_{n}^{[I-1,J,K,L+1]}(x_1,\ldots,x_{I-1};y_1,\ldots,y_J;v_1,\ldots,v_K;w_1,\ldots,w_{L+1})
\\ =  C_0\int_{-\infty}^{\infty}\!\int_{-\infty}^{\infty}\!
\frac{\prod_{k=1}^K(x_I-v_k)}{\prod_{i=1}^{I-1}(x_{I}-x_i)}\frac{\til
w(x_I,\eta)}{w_{L+1}-\eta} \\ \times
\det\begin{pmatrix}
\K_{1,1}(x_1,v_1) & \ldots & \K_{1,1}(x_I,v_1) & \K_{2,1}(w_1,v_1) & \ldots & \K_{2,1}(w_L,v_1) \\
\vdots & & \vdots & \vdots & & \vdots\\
\K_{1,1}(x_1,v_K) & \ldots & \K_{1,1}(x_I,v_K) & \K_{2,1}(w_1,v_K) & \ldots & \K_{2,1}(w_L,v_K) \\
\K_{1,2}(x_1,y_1) & \ldots & \K_{1,2}(x_I,y_1) & \K_{2,2}(w_1,y_1) & \ldots & \K_{2,2}(w_L,y_1) \\
\vdots & & \vdots & \vdots & & \vdots \\
\K_{1,2}(x_1,y_J) & \ldots & \K_{1,2}(x_I,y_J) & \K_{2,2}(w_1,y_J) & \ldots & \K_{2,2}(w_L,y_J) \\
p_{n-1+J-L}(x_1) & \ldots & p_{n-1+J-L}(x_I) & \P_{n-1+J-L}(w_1) & \ldots & \P_{n-1+J-L}(w_L) \\
\vdots & & \vdots & \vdots & & \vdots \\
p_{n-2+I-K}(x_1) & \ldots & p_{n-2+I-K}(x_I) & \P_{n-2+I-K}(w_1) & \ldots &
\P_{n-2+I-K}(w_L)
\end{pmatrix}dx_I\ d\eta,
\end{multline}
where $C_0$ is a new constant which depends on each of the numbers
$x_i,y_j,v_k,w_l$, but \emph{not} on $w_{L+1}$. We will work this out. Note
that the factor $\frac{\prod_{i=1}^{I-1}(x_{I}-x_i)}{\prod_{k=1}^{K}(x_I-v_k)}$
in the definition of $\til w(x_I,\eta)$ in \eqref{wtilUv} cancels with the
prefactor in the integrand of \eqref{proof:gen0}. By linearity, we can then
apply the double integration entrywise inside the $I$th column of the matrix in
\eqref{proof:gen0}. Then the entries in the $I$th column of \eqref{proof:gen0}
transform into expressions of the form \begin{equation}\label{proofgen:1}
\int_{-\infty}^{\infty}\!\int_{-\infty}^{\infty}\!
\K_{1,1}(x_I,v_k)\frac{\prod_{j=1}^{J}(\eta-y_j)}{\prod_{l=1}^{L}(\eta-w_l)}\frac{1}{w_{L+1}-\eta}e^{-V(x_I)-W(\eta)+\tau
x_I \eta}\ dx_I\ d\eta,
\end{equation}
for $k=1,\ldots,K$,
\begin{equation}\label{proofgen:2} \int_{-\infty}^{\infty}\!\int_{-\infty}^{\infty}\!
\K_{1,2}(x_I,y_j)\frac{\prod_{j=1}^{J}(\eta-y_j)}{\prod_{l=1}^{L}(\eta-w_l)}\frac{1}{w_{L+1}-\eta}e^{-V(x_I)-W(\eta)+\tau
x_I \eta}\ dx_I\ d\eta,
\end{equation}
for $j=1,\ldots,J$, and \begin{equation}\label{proofgen:3}
\int_{-\infty}^{\infty}\!\int_{-\infty}^{\infty}\!
p_{n-1+p}(x_I)\frac{\prod_{j=1}^{J}(\eta-y_j)}{\prod_{l=1}^{L}(\eta-w_l)}\frac{1}{w_{L+1}-\eta}e^{-V(x_I)-W(\eta)+\tau
x_I \eta}\ dx_I\ d\eta,
\end{equation}
for $p=J-L,\ldots,I-K-1$.

Next, we substitute the partial fraction decomposition \eqref{pfd1} in
\eqref{proofgen:1} and \eqref{proofgen:3}, and we substitute \eqref{pfd2} in
\eqref{proofgen:2}. Doing this for each of the entries in the $I$th column of
\eqref{proof:gen0}, the determinant can be split in a sum of three terms,
\begin{equation}\label{threeterms} D_1+D_2+D_3,
\end{equation}
corresponding to the three terms on the right hand sides of
\eqref{pfd1}--\eqref{pfd2}.

For the first term $D_1$ (which is obtained by selecting the terms
$\frac{c_{L+1}}{\eta-w_{L+1}}$ and
$\frac{c_{L+1}(\eta-y_j)}{(w_{L+1}-y_j)(\eta-w_{L+1})}$ in \eqref{pfd1} and
\eqref{pfd2} respectively), the expressions
\eqref{proofgen:1}--\eqref{proofgen:3} transform into
 \begin{equation}\label{proofgen:4}
c_{L+1}\int_{-\infty}^{\infty}\!\int_{-\infty}^{\infty}\!
\K_{1,1}(x_I,v_k)\frac{1}{w_{L+1}-\eta}e^{-V(x_I)-W(\eta)+\tau x_I \eta}\ dx_I\
d\eta,
\end{equation}
for $k=1,\ldots,K$,
\begin{equation}\label{proofgen:5} c_{L+1}\int_{-\infty}^{\infty}\!\int_{-\infty}^{\infty}\!
\K_{1,2}(x_I,y_j)\frac{\eta-y_j}{(w_{L+1}-y_j)(w_{L+1}-\eta)}e^{-V(x_I)-W(\eta)+\tau
x_I \eta}\ dx_I\ d\eta,
\end{equation}
for $j=1,\ldots,J$, and
\begin{equation}\label{proofgen:6}
c_{L+1}\int_{-\infty}^{\infty}\!\int_{-\infty}^{\infty}\!
p_{n-1+p}(x_I)\frac{1}{w_{L+1}-\eta}e^{-V(x_I)-W(\eta)+\tau x_I \eta}\ dx_I\
d\eta,
\end{equation}
for $p=J-L,\ldots,I-K-1$. Comparing this with
Lemma~\ref{lemma:kernel:integrate}, we see that
\eqref{proofgen:4}--\eqref{proofgen:6} are nothing but $c_{L+1}$ times
$\K_{2,1}(w_{L+1},v_k)$, $\K_{2,2}(w_{L+1},y_j)$ and $\P_{n-1+p}(w_{L+1})$
respectively. The factor $c_{L+1}$ can be taken out of the $I$th column and be
put in front of the determinant. Note that
$$ c_{L+1}= \frac{\prod_{j=1}^J (w_{L+1}-y_j)}{\prod_{l=1}^L (w_{L+1}-w_l)},
$$
a fact which is easily checked from \eqref{pfd1}.

For the second term $D_2$ in \eqref{threeterms} (which is obtained by selecting
the terms $\sum_{l=1}^L \frac{c_l}{\eta-w_l}$ and $\sum_{l=1}^L
\frac{c_l(\eta-y_j)}{(w_l-y_j)(\eta-w_l)}$ in \eqref{pfd1}--\eqref{pfd2}
respectively), we obtain similarly that \eqref{proofgen:1}--\eqref{proofgen:3}
transform into a linear combination of $\K_{2,1}(w_{l},v_k)$,
$\K_{2,2}(w_{l},y_j)$ and $\P_{n-1+p}(w_{l})$, $l=1,\ldots,L$. But then the
$I$th column in \eqref{proof:gen0} is a linear combination of columns
$I+1,\ldots,I+L$ and so the determinant $D_2$ vanishes.

Finally, the third determinant $D_3$ (corresponding to the polynomial parts
$P(\eta)$ and $P(\eta)-P(y_j)$ in \eqref{pfd1}--\eqref{pfd2} respectively)
vanishes as well, by virtue of Lemmas~\ref{lemma:K12repr} and
\ref{lemma:vanishing}.

Summarizing, we showed that \eqref{proof:gen0} equals
\begin{multline}\label{gencase:proof}
P_{n}^{[I-1,J,K,L+1]}(x_1,\ldots,x_{I-1};y_1,\ldots,y_J;v_1,\ldots,v_K;w_1,\ldots,w_{L+1})
 =  C_1
\frac{\prod_{j=1}^J(w_{L+1}-y_j)}{\prod_{l=1}^L(w_{L+1}-w_l)}\\
\times\det\begin{pmatrix}
\K_{1,1}(x_1,v_1) & \ldots & \K_{1,1}(x_{I-1},v_1) & \K_{2,1}(w_1,v_1) & \ldots & \K_{2,1}(w_{L+1},v_1) \\
\vdots & & \vdots & \vdots & & \vdots\\
\K_{1,1}(x_1,v_K) & \ldots & \K_{1,1}(x_{I-1},v_K) & \K_{2,1}(w_1,v_K) & \ldots & \K_{2,1}(w_{L+1},v_K) \\
\K_{1,2}(x_1,y_1) & \ldots & \K_{2,1}(x_{I-1},y_1) & \K_{2,2}(w_1,y_1) & \ldots & \K_{2,2}(w_{L+1},y_1) \\
\vdots & & \vdots & \vdots & & \vdots \\
\K_{1,2}(x_1,y_J) & \ldots & \K_{2,1}(x_{I-1},y_J) & \K_{2,2}(w_1,y_J) & \ldots & \K_{2,2}(w_{L+1},y_J) \\
p_{n-1+J-L}(x_1) & \ldots & p_{n-1+J-L}(x_{I-1}) & \P_{n-1+J-L}(w_1) & \ldots & \P_{n-1+J-L}(w_{L+1}) \\
\vdots & & \vdots & \vdots & & \vdots \\
p_{n-2+I-K}(x_1) & \ldots & p_{n-2+I-K}(x_{I-1}) & \P_{n-2+I-K}(w_1) & \ldots &
\P_{n-2+I-K}(w_{L+1})
\end{pmatrix},
\end{multline}
where $C_1=\pm C_0$ depends on each of the numbers $x_i,y_j,v_k,w_l$, but
\emph{not} on $w_{L+1}$.

To find $C_1$, we compute the leading order behavior of \eqref{gencase:proof}
in $w_{L+1}$; this can be done by replacing each of the entries in the last
column of \eqref{gencase:proof} by their dominant term in $w_{L+1}$ as in
Lemma~\ref{lemma:leadingterms}. (More precisely, if $I-K>J-L$ then only the
entry $\P_{n-1+J-L}(w_{L+1})$ in \eqref{gencase:proof} contributes to the
dominant term, while if $I-K=J-L$ then there is a contribution from \emph{all}
entries in the last column of \eqref{gencase:proof}.) Taking into account that
\begin{multline*} 
P_{n}^{[I-1,J,K,L+1]}(x_1,\ldots,x_{I-1};y_1,\ldots,y_J;v_1,\ldots,v_K;w_1,\ldots,w_{L+1})
\\=
w_{L+1}^{-n}P_{n}^{[I-1,J,K,L]}(x_1,\ldots,x_{I-1};y_1,\ldots,y_J;v_1,\ldots,v_K;w_1,\ldots,w_{L})+O\left(w_{L+1}^{-n-1}\right),
\end{multline*}
as $w_{L+1}\to\infty$, a fact which trivially follows from \eqref{avcharpol},
we then obtain an expression for $C_1$.

In a similar way we can compute the leading order behavior of \eqref{gencase:a}
in $x_I$, i.e., we can write $P_{n}^{[I,J,K,L]}$ in the form
\begin{multline}\label{gencase:proof2}
C_2\frac{\prod_{k=1}^K(x_{I}-v_k)}{\prod_{i=1}^{I-1}(x_{I}-x_i)}\\
\times\det\begin{pmatrix}
\K_{1,1}(x_1,v_1) & \ldots & \K_{1,1}(x_{I},v_1) & \K_{2,1}(w_1,v_1) & \ldots & \K_{2,1}(w_{L},v_1) \\
\vdots & & \vdots & \vdots & & \vdots\\
\K_{1,1}(x_1,v_K) & \ldots & \K_{1,1}(x_{I},v_K) & \K_{2,1}(w_1,v_K) & \ldots & \K_{2,1}(w_{L},v_K) \\
\K_{1,2}(x_1,y_1) & \ldots & \K_{2,1}(x_{I},y_1) & \K_{2,2}(w_1,y_1) & \ldots & \K_{2,2}(w_{L},y_1) \\
\vdots & & \vdots & \vdots & & \vdots \\
\K_{1,2}(x_1,y_J) & \ldots & \K_{2,1}(x_{I},y_J) & \K_{2,2}(w_1,y_J) & \ldots & \K_{2,2}(w_{L},y_J) \\
p_{n+J-L}(x_1) & \ldots & p_{n+J-L}(x_{I}) & \P_{n+J-L}(w_1) & \ldots & \P_{n+J-L}(w_{L}) \\
\vdots & & \vdots & \vdots & & \vdots \\
p_{n+I-K-1}(x_1) & \ldots & p_{n+I-K-1}(x_{I}) & \P_{n+I-K-1}(w_1) & \ldots &
\P_{n+I-K-1}(w_{L})
\end{pmatrix},
\end{multline}
and then we can replace each of the entries in the $I$th column of
\eqref{gencase:proof2} by their dominant term in $x_I$ as in
Lemma~\ref{lemma:leadingterms}. (More precisely, if $I-K>J-L$ then only the
entry $p_{n+I-K-1}(x_{I})$ in \eqref{gencase:proof2} contributes to the
dominant term, while if $I-K=J-L$ then there is a contribution from \emph{all}
entries in the $I$th column of \eqref{gencase:proof2}.) Taking into account
that
\begin{multline*} 
P_{n}^{[I,J,K,L]}(x_1,\ldots,x_{I};y_1,\ldots,y_J;v_1,\ldots,v_K;w_1,\ldots,w_{L})
\\=
x_{I}^{n}P_{n}^{[I-1,J,K,L]}(x_1,\ldots,x_{I-1};y_1,\ldots,y_J;v_1,\ldots,v_K;w_1,\ldots,w_{L})+O\left(x_{I}^{n-1}\right),
\end{multline*}
as $x_I\to\infty$, we then obtain an expression for  $C_2$ in
\eqref{gencase:proof2}. A straightforward calculation now shows that the so
obtained expressions for $C_1$ and $C_2$ in \eqref{gencase:proof} and
\eqref{gencase:proof2} are equal to each other, up to a shift $n\mapsto n-1$ of
the index, and a factor $(-1)^{I+J+K+1}/h_{n+J-L-1}^2$. So the constants $C_1$
and $C_2$ are related precisely in the way that is required for applying the
transformation \eqref{stepinduc1} to the prefactor in \eqref{gencase:a}.

Summarizing, we have shown that \eqref{gencase:a} is compatible with any
transformation of the form \eqref{stepinduc1}. In a similar way one shows this
for transformations of the form \eqref{stepinduc2}. Then applying the
transformations \eqref{stepinduc1}--\eqref{stepinduc2} repeatedly, and using
Theorem~\ref{theorem:avpol:num} as induction basis, we obtain \eqref{gencase:a}
in its full generality. $\bol$

\begin{remark} An alternative proof of Theorem~\ref{theorem:avpol:gen} can be obtained
by establishing formulas of Christoffel-Uvarov type \cite{BDS,Sz,Uv}. To this
end one can use similar ideas as in the proof above. We preferred the given
proof since it is probably shorter.
\end{remark}

\section{Applications}
\label{section:appl}

In this section we give two applications of our results. Here we follow ideas
for the $1$-matrix model.

\subsection{Generating function for averages of products of traces}
\label{subsection:traces}

In this section we show how the results in this paper allow to obtain
expectation values of products of traces of the form
\begin{equation*}\frac{1}{Z_n}
\int\!\int\!\left(\prod_{i=1}^I \textrm{Tr} (M_1^{m_i})\prod_{j=1}^J
\textrm{Tr} (M_2^{n_j})\right) e^{\textrm{Tr}(-V(M_1)-W(M_2)+\tau M_1M_2)}\
dM_1\ dM_2,
\end{equation*}
or equivalently
\begin{multline}\label{traces2}\frac{1}{\til Z_n}\int_{-\infty}^{\infty}\!\ldots\int_{-\infty}^{\infty}
\prod_{i=1}^I\left(\sum_{k=1}^n
\lam_k^{m_i}\right)\prod_{j=1}^J\left(\sum_{k=1}^n
\mu_k^{n_j}\right)\\
\times\prod_{i=1}^n  \left( e^{-V(\lam_i)}
e^{-W(\mu_i)}\right)\Delta(\mathbf{\lam})\Delta(\mathbf{\mu})\det(e^{\tau\lam_i\mu_j})_{i,j=1}^n
\prod_{i=1}^n \left( d\lam_i\ d\mu_i\right),
\end{multline}
with exponents $m_1,\ldots,m_I;n_1,\ldots,n_J\in\enn\cup\{0\}$, see
\eqref{inteigs}. Here we follow Berg\`ere \cite{Bergere2}. We start from
\eqref{avcharpol:Harish} with $I=K$ and $J=L$:
\begin{multline}\label{traces3}\frac{1}{\til Z_n}\int_{-\infty}^{\infty}\!\ldots\int_{-\infty}^{\infty}
\prod_{i=1}^I\left(\frac{\prod_{k=1}^n (x_i-\lam_k)}{\prod_{k=1}^n
(v_i-\lam_k)}\right)\prod_{j=1}^J\left(\frac{\prod_{k=1}^n
(y_j-\lam_k)}{\prod_{k=1}^n (w_j-\lam_k)}\right)\\
\times\prod_{i=1}^n  \left( e^{-V(\lam_i)}
e^{-W(\mu_i)}\right)\Delta(\mathbf{\lam})\Delta(\mathbf{\mu})\det(e^{\tau\lam_i\mu_j})_{i,j=1}^n
\prod_{i=1}^n \left( d\lam_i\ d\mu_i\right).
\end{multline}
Applying the operator \begin{equation}\label{operatordiff}\prod_{i=1}^I
\left(\frac{\partial}{\partial x_i}\right)_{v_i=x_i}\prod_{j=1}^J
\left(\frac{\partial}{\partial y_i}\right)_{w_j=y_j}\end{equation} to
\eqref{traces3} leads to the (multi-variate) Cauchy transform
\begin{multline}\label{traces4}\frac{1}{\til Z_n}\int_{-\infty}^{\infty}\!\ldots\int_{-\infty}^{\infty}
\prod_{i=1}^I\left(\sum_{k=1}^n
\frac{1}{x_i-\lam_k}\right)\prod_{j=1}^J\left(\sum_{k=1}^n
\frac{1}{y_j-\mu_k}\right)\\
\times\prod_{i=1}^n  \left( e^{-V(\lam_i)}
e^{-W(\mu_i)}\right)\Delta(\mathbf{\lam})\Delta(\mathbf{\mu})\det(e^{\tau\lam_i\mu_j})_{i,j=1}^n
\prod_{i=1}^n \left( d\lam_i\ d\mu_i\right),
\end{multline}
where we assume that all $x_i,y_j\in\cee\setminus\er$. Note that the large
$x_i,y_j$ expansion of \eqref{traces4} is a formal power series whose
coefficients are the expectation values of products of traces \eqref{traces2}.

On the other hand, Theorem~\ref{theorem:avpol:gen} with $I=K$ and $J=L$ shows
that \eqref{traces3} equals
\begin{multline}\label{traces5}
= \frac{\prod_{i,k=1}^I(x_i-v_k)\prod_{j,l=1}^J(y_j-w_l)}{\prod_{1\leq i<j\leq
I}(x_j-x_i)(v_i-v_j)
\prod_{1\leq i<j\leq J}(y_j-y_i)(w_i-w_j)}\\
\times \det\begin{pmatrix}
\K_{1,1}(x_1,v_1) & \ldots & \K_{1,1}(x_1,v_I) & \K_{1,2}(x_1,y_1) & \ldots & \K_{1,2}(x_1,y_J) \\
\vdots & & \vdots & \vdots & & \vdots \\
\K_{1,1}(x_I,v_1) & \ldots & \K_{1,1}(x_I,v_I) & \K_{1,2}(x_I,y_1) & \ldots & \K_{1,2}(x_I,y_J) \\
\K_{2,1}(w_1,v_1) & \ldots & \K_{2,1}(w_1,v_I) & \K_{2,2}(w_1,y_1) & \ldots & \K_{2,2}(w_1,y_J) \\
\vdots & & \vdots & \vdots & & \vdots\\
\K_{2,1}(w_J,v_1) & \ldots & \K_{2,1}(w_J,v_I) & \K_{2,2}(w_J,y_1) & \ldots &
\K_{2,2}(w_J,y_J)
\end{pmatrix}.
\end{multline}
Applying the operator \eqref{operatordiff} to this determinant, and making a
small calculation, one gets
\begin{equation}\label{traces6}
\det\begin{pmatrix}
\widehat K_{1,1}(x_1,x_1) & \ldots & \K_{1,1}(x_1,x_I) & \K_{1,2}(x_1,y_1) & \ldots & \K_{1,2}(x_1,y_J) \\
\vdots & & \vdots & \vdots & & \vdots \\
\K_{1,1}(x_I,x_1) & \ldots & \widehat K_{1,1}(x_I,x_I) & \K_{1,2}(x_I,y_1) & \ldots & \K_{1,2}(x_I,y_J) \\
\K_{2,1}(y_1,x_1) & \ldots & \K_{2,1}(y_1,x_I) & \widehat K_{2,2}(y_1,y_1) & \ldots & \K_{2,2}(y_1,y_J) \\
\vdots & & \vdots & \vdots & & \vdots\\
\K_{2,1}(y_J,x_1) & \ldots & \K_{2,1}(y_J,x_I) & \K_{2,2}(y_J,y_1) & \ldots &
\widehat K_{2,2}(y_J,y_J)
\end{pmatrix},
\end{equation}
where the entries on the main diagonal are defined by $$\widehat K_{1,1}(x,x) =
\sum_{i=0}^{n-1}\frac{1}{h_i^2}p_i(x)\Q_i(x),$$ and
$$\widehat K_{2,2}(y,y) =
\sum_{i=0}^{n-1}\frac{1}{h_i^2}\P_i(y)q_i(y),$$ respectively; compare with
\eqref{def:K11sum} and \eqref{def:K22sum}. So \eqref{traces6} yields a
determinantal formula for the multi-variate Cauchy transform \eqref{traces4}.
By the above discussion, this is also the generating function for averages of
products of traces \eqref{traces2}.

\subsection{Eynard-Mehta theorem for correlation functions}
\label{subsection:EynardMehta}

In this section we show how Theorem~\ref{theorem:avpol:gen} can be used to give
an alternative proof of the Eynard-Mehta theorem for correlation functions in
the two-matrix model. Alternative proofs of this formula, which are applicable
for the more general model of random matrices coupled in a chain, can be found
in \cite{BR,EM,Joh,NF,TW}.

The \emph{correlation function} $R_{I,J}$ is defined from \eqref{inteigs} by
\begin{multline}\label{EM1} R_{I,J}(\lam_1,\ldots,\lam_{I};\mu_1,\ldots,\mu_{J})
= \frac{n!}{(n-I)!}\frac{n!}{(n-J)!}
\\ \times\frac{1}{\til
Z_n}\int_{-\infty}^{\infty}\!\ldots\int_{-\infty}^{\infty} \prod_{i=1}^n \left(
e^{-V(\lam_i)}
e^{-W(\mu_i)}\right)\Delta(\mathbf{\lam})\Delta(\mathbf{\mu})\det(e^{\tau
\lam_i \mu_j})_{i,j=1}^n \prod_{i=I+1}^{n} d\lam_i \prod_{j=J+1}^{n} d\mu_j.
\end{multline}
The \emph{Eynard-Mehta theorem} \cite{EM} then asserts that
\begin{multline}\label{EM2}
R_{I,J}(\lam_1,\ldots,\lam_{I};\mu_1,\ldots,\mu_{J})
\\ = \det\begin{pmatrix}
K_{1,1}(\lam_1,\lam_1) & \ldots & K_{1,1}(\lam_1,\lam_I) & K_{1,2}(\lam_1,\mu_1) & \ldots & K_{1,2}(\lam_1,\mu_J) \\
\vdots & & \vdots & \vdots & & \vdots\\
K_{1,1}(\lam_I,\lam_1) & \ldots & K_{1,1}(\lam_I,\lam_I) & K_{1,2}(\lam_I,\mu_1) & \ldots & K_{1,2}(\lam_I,\mu_J)\\
K_{2,1}(\mu_1,\lam_1) & \ldots & K_{2,1}(\mu_1,\lam_I) & K_{2,2}(\mu_1,\mu_1) & \ldots & K_{2,2}(\mu_1,\mu_J)\\
\vdots & & \vdots & \vdots & & \vdots \\
K_{2,1}(\mu_J,\lam_1) & \ldots & K_{2,1}(\mu_J,\lam_I) & K_{2,2}(\mu_J,\mu_1) & \ldots & K_{2,2}(\mu_J,\mu_J) \\
\end{pmatrix}.
\end{multline}

We now establish this formula by using the formulas in
Section~\ref{subsection:traces}. Here we follow \cite{BS}. We start from the
multi-variate Cauchy transform in \eqref{traces4}. Applying the operators
$$ f\mapsto \frac{1}{2\pi i}\lim_{\epsilon\to 0+}\left(f|_{x_i=\lam_i-\epsilon i} - f|_{x_i=\lam_i+\epsilon i}\right)$$
subsequently for $i=1,\ldots,I$ and then
$$ f\mapsto \frac{1}{2\pi i}\lim_{\epsilon\to 0+}\left(f|_{y_j=\mu_j-\epsilon i} - f|_{y_i=\mu_i+\epsilon i}\right)$$
for $j=1,\ldots,J$ to \eqref{traces4}, we obtain by the Stieltjes-Perron
inversion principle \cite{Sz} precisely the correlation function \eqref{EM1}.
On the other hand, applying these same operations to \eqref{traces6} and using
again the Stieltjes-Perron inversion principle leads to the right hand side of
\eqref{EM2}. This establishes \eqref{EM2}.

\section{Concluding remarks}
\label{section:chain}

The two-matrix model is a particular instance of a more general model,
sometimes referred to as \emph{random matrices coupled in a chain}, see e.g.\
\cite{AvMV,BR,EM,Joh,Mehta,NF,TW}. The Eynard-Mehta theorem for correlation
functions can be formulated for this more general model.

In view of this observation, it is natural to ask whether the results in this
paper can be extended to the more general model of random matrices coupled in a
chain. That is, one may ask whether the averages of products and ratios of
characteristic polynomials in this model can still be written as determinants
built out of (transformed) Eynard-Mehta kernels and biorthogonal polynomials. A
little thought reveals that such a result, if it exists, should be a
non-trivial extension of Theorem~\ref{theorem:avpol:gen}, except maybe for
special configurations of the external sources. This is an open problem.

Another question of interest is whether the results in this paper have an
analogue for the Cauchy two-matrix model in \cite{BGS}.

\end{document}